\documentclass{article}
\usepackage{arxiv}

\usepackage[utf8]{inputenc} % allow utf-8 input
\usepackage[T1]{fontenc}    % use 8-bit T1 fonts
\usepackage{hyperref}       % hyperlinks
\usepackage{url}            % simple URL typesetting
\usepackage{booktabs}       % professional-quality tables
\usepackage{amsfonts}       % blackboard math symbols
\usepackage{nicefrac}       % compact symbols for 1/2, etc.
\usepackage{microtype}      % microtypography
\usepackage{lipsum}

\usepackage[small]{caption}

\usepackage{amssymb}
\usepackage[round]{natbib}
\usepackage{mathtools}
\usepackage{amsmath}
\usepackage{tikz}
\usetikzlibrary{calc}
\usepackage{todonotes}
\usepackage{dsfont}
\usepackage{enumitem}
\usepackage{tikz}
\usepackage{amsthm}
\usepackage{thm-restate}
\usepackage{booktabs}
\usepackage{newfloat}
\usepackage{wrapfig}
\usepackage{dsfont}
\usepackage{comment}
\usepackage{multirow}
\usepackage{algorithm}
\usepackage{algorithmic}
\usepackage{bbm}

\newtheorem{definition}{Definition}

\DeclareMathOperator*{\argmax}{arg\,max}

\newcommand{\N}{\mathcal{N}}
\newcommand{\M}{\mathcal{M}}
\newcommand{\V}{\mathcal{V}}

\newcommand{\sset}{\mathcal{S}}

\newcommand*\circled[1]{\tikz[baseline=(char.base)]{
		\node[shape=circle,draw,inner sep=1pt] (char) {\scriptsize #1};}}

\title{Public Signaling in Bayesian Ad Auctions}

\author{Francesco Bacchiocchi$^{\dagger}$, Matteo Castiglioni$^{\S}$, Alberto Marchesi$^{\S}$, Giulia Romano$^{\S}$, Nicola Gatti$^{\S}$ \\
	Politecnico di Milano, Piazza Leonardo da Vinci 32, I-20133, Milan, Italy \\
	 $^{\dagger}$\texttt{francesco.bacchiocchi@mail.polimi.it}, $^{\S}$\texttt{\{name.surname\}@polimi.it} }
\begin{document}
\maketitle
\begin{abstract}
  We study \emph{signaling} in \emph{Bayesian ad auctions}, in which bidders' valuations depend on a random, unknown state of nature.
  The auction mechanism has complete knowledge of the actual state of nature, and it can send signals to bidders so as to disclose information about the state and increase revenue.
  For instance, a state may collectively encode some features of the user that are known to the mechanism only, since the latter has access to data sources unaccessible to the bidders.
  We study the problem of computing \emph{how the mechanism should send signals to bidders in order to maximize revenue}.
  While this problem has already been addressed in the easier setting of second-price auctions, to the best of our knowledge, our work is the first to explore ad auctions with more than one slot.
  In this paper, we focus on \emph{public} signaling and VCG mechanisms, under which bidders truthfully report their valuations.
  We start with a negative result, showing that, in general, the problem does \emph{not} admit a PTAS unless $\mathsf{P} = \mathsf{NP}$, even when bidders' valuations are known to the mechanism.
  The rest of the paper is devoted to settings in which such negative result can be circumvented.
  First, we prove that, with \emph{known valuations}, the problem can indeed be solved in polynomial time when either the number of states $d$ or the number of slots $m$ is fixed.
  Moreover, in the same setting, we provide an FPTAS for the case in which bidders are \emph{single minded}, but $d$ and $m$ can be arbitrary.
  Then, we switch to the \emph{random valuations} setting, in which these are randomly drawn according to some probability distribution.
  In this case, we show that the problem admits an FPTAS, a PTAS, and a QPTAS, when, respectively, $d$ is fixed, $m$ is fixed, and bidders' valuations are bounded away from zero. 
\end{abstract}

\section{Introduction}\label{sec:intro}

Nowadays, worldwide spending in digital advertising is skyrocketing, and this growth is primarily driven by \emph{ad auctions}.
These account for almost all market share, since they are at the core of popular advertising platforms, such as, \emph{e.g.}, those by Google, Amazon, and Facebook.
According to a recent report by~\citeauthor{eMarketer}~[\citeyear{eMarketer}], digital ad spending will reach over \$490 billion in 2021 and zoom past half a trillion in 2022.

We study \emph{signaling} in ad auction settings by means of the \emph{Bayesian persuasion} framework~\cite{kamenica2011bayesian}.
Over the last years, this framework has received considerable attention from the computer science community, due to its applicability to many real-world scenarios, such as, \emph{e.g.}, online advertising~\cite{bro2012send,emek2014signaling,badanidiyuru2018targeting}, voting~\cite{alonso2016persuading,cheng2015mixture,castiglioni2019persuading,semipublic}, traffic routing~\cite{vasserman2015implementing,bhaskar2016hardness,castiglioni2020signaling}, recommendation systems~\cite{mansour2016bayesian}, security~\cite{rabinovich2015information,xu2016signaling}, and product marketing~\cite{babichenko2017algorithmic,candogan2019persuasion}.

In a standard ad auction, the advertisers (also called bidders) compete for displaying their ads on a limited number of slots, and each bidder has their own private valuation representing how much they value a click on their ad. 
In this work, we study \emph{Bayesian ad auctions}, which are characterized by the fact that bidders' valuations depend on a random, unknown state of nature.
The auction mechanism has complete knowledge of the actual state of nature, and it can send signals to bidders so as to disclose information about the state and increase revenue.
In particular, the auction mechanism \emph{commits to a signaling scheme}, which is defined as a randomized mapping from states of nature to signals being sent to the bidders.
Our model fits many real-world applications that are \emph{not} captured by classical ad auctions.
For instance, a state of nature may collectively encode some features of the user visualizing the ads---such as, \emph{e.g.}, age, gender, or geographical region---that are known to the mechanism only, since the latter has access to data sources unaccessible to the bidders.

We study the problem of computing \emph{a revenue-maximizing signaling scheme for the mechanism}.
In particular, in this paper we focus on \emph{public} signaling, in which the mechanism can only send a single signal that is observed by all the bidders.
Moreover, we restrict our attention to VCG mechanisms, which are widely used in practice and have the appealing property of inducing bidders to truthfully report their valuations.
While the signaling problem studied in this paper has already been addressed in the easier setting of second-price auctions~\cite{badanidiyuru2018targeting}, to the best of our knowledge, our work is the first to explore algorithmic signaling in general ad auctions with more than one slot.

\subsection{Original Contributions}

We start our analysis with a negative result, showing that, in general, the revenue-maximizing problem with public signaling does \emph{not} admit a PTAS unless $\mathsf{P} = \mathsf{NP}$, even when bidders' valuations are known to the mechanism.
Thus, in the rest of the paper, we address settings in which we can prove that such a negative result can be circumvented.

First, we show that, in the \emph{known valuations} setting, the problem admits a polynomial-time algorithm when either the number of slots $m$ or the number of states $d$ is fixed.
The proposed algorithms work by solving suitably-defined \emph{linear programs} (LPs) of polynomial size, thanks to the crucial property that, when either $m$ or $d$ is fixed, there always exists an optimal signaling scheme using a polynomial number of different signals. 
Moreover, we also study special instances in which the bidders are \emph{single minded}, but $m$ and $d$ can be arbitrary.
In this case, each bidder positively values a click on their ad only when the actual state of nature is a specific (single) state, and all the bidders interested in the same state value a click on their ad for the same amount.
By exploiting a particular combinatorial structure of the set of bidders' posterior distributions induced by signaling schemes, we are able to provide an FPTAS in such setting.
The algorithm works by applying the ellipsoid method in a non-trivial way, with only access to an approximate polynomial-time separation oracle.
The latter is implemented by a rather involved dynamic programming algorithm, which works thanks to the particular structure of the set of bidders' posteriors.

Then, we switch the attention to the \emph{random valuations} setting, where bidders' valuations are unknown to the mechanism, but randomly drawn according to some probability distribution.
In this case, we first provide some preliminary results that establish useful connections between the optimal value of the revenue-maximizing problem and that of optimal signaling schemes restricted to suitably-defined finite sets of posterior distributions.
These sets are defined so that the expected revenue of the mechanism is \emph{``stable''}, meaning that it does \emph{not} decrease too much when restricting signaling schemes to use posteriors in such sets.
In particular, for our results we use sets of $q$-uniform posteriors, for suitable values of $q$.
As a preliminary step, we also show that it is possible to compute an approximately-optimal signaling scheme having only access to a finite number of samples from the distribution of bidders' valuations.
In conclusion, all the preliminary results described so far allow us to prove that, in the random valuations setting, the problem admits an FPTAS, a PTAS, and a QPTAS, when, respectively, $d$ is fixed, $m$ is fixed, and bidders' valuations are bounded away from zero.\footnote{All the proofs are in the Supplementary Material.}

\subsection{Related Works}

To the best of our knowledge, the algorithmic study of signaling in auctions is limited to the \emph{second-price auction}, which can be seen as a special ad auction with a single slot.

\citeauthor{emek2014signaling}~[\citeyear{emek2014signaling}] study second-price auctions in the known valuations setting.
% in which the sender knows the buyers' valuations.
They provide an LP to compute an optimal public signaling scheme.
Moreover, they show that it is $\mathsf{NP}$-hard to compute an optimal signaling scheme in the random valuations setting.
In our work, we generalize their positive result, in order to provide our polynomial-time algorithm working when the number of slots $m$ is fixed.
% and the valuations are known.

\citeauthor{cheng2015mixture}~[\citeyear{cheng2015mixture}] complement the hardness result of~\cite{emek2014signaling} by providing a PTAS for the random valuations setting.
This result cannot be extended to ad auctions, as we show in our first negative result.
% since we show that, unless $\mathsf{P} = \mathsf{NP}$, there is not a PTAS even when the valuations are known.
%
However, we provide two generalizations of the result by~\citeauthor{cheng2015mixture}~[\citeyear{cheng2015mixture}]: we provide a PTAS for the random valuations setting with a fixed number of slots $m$, and a QPTAS when the bidder's valuations are bounded away from zero.

Finally, \citeauthor{badanidiyuru2018targeting}~[\citeyear{badanidiyuru2018targeting}] study algorithms whose running time does \emph{not} depend on the number of states of nature.
Moreover, they initiate the study of private signaling schemes, showing that, in second-price auctions, private signaling introduces non-trivial equilibrium selection problems.

\section{Preliminaries}\label{sec:preliminaries}

In a standard \emph{ad auction} (see also the book by~\citeauthor{nisan2001algorithmic}~[\citeyear{nisan2001algorithmic}] for more details), there is a set $\N \coloneqq \{ 1, \ldots, n\}$ of \emph{advertisers} (or \emph{bidders}) who compete for displaying their ads on a set $\M \coloneqq \{ 1, \ldots, m \}$ of \emph{slots}, with $m \leq n$.
Each bidder $i \in \N$ is characterized by a \emph{private valuation} $ v_i \in [0,1]$, which represents how much they value a click on their ad.
%, and by an ad \emph{quality} $q_i \in [0,1]$ encoding the probability with which the ad is clicked once observed.
%
%For the ease of notation, we let $v_i = q_i \tilde v_i$ for all $i \in \N$.
%
Moreover, each slot $j \in \M$ is associated with a \emph{click through rate} parameter $\lambda_j \in[0,1]$, which is the probability with which the slot is clicked by a user.\footnote{In this work, for the ease of presentation, we assume that the click through rate only depends on the slot and \emph{not} on the ad being displayed. In general, each slot may have its own \emph{prominence} value---the probability with which a user observes it---and each bidder may have their own ad \emph{quality}---the probability with which their ad is clicked once observed---, so that the click through rate is defined as the product of these two quantities. All the results in this paper can be easily extended to such general model.} 
W.l.o.g., we assume that the sots are ordered so that $\lambda_1 \geq \ldots \geq \lambda_m$.
%we assume that bidders and sots are ordered so that $v_1 \geq \ldots \geq v_n$ and $\lambda_1 \geq \ldots \lambda_m$.
%
The auction goes on as follows: first, each bidder $i \in \N$ separately reports a bid $b_i \in [0,1]$ to the auction mechanism; then, based on the bids, the latter allocates an ad to each slot and defines how much each bidder has to pay the mechanism for a click on their ad.
We focus on \emph{truthful} mechanisms, and the VCG mechanism in particular (see the book by~\citeauthor{mas1995microeconomic}~[\citeyear{mas1995microeconomic}] for a complete description of the mechanism).
%, which, in the ad auction setting, works as follows: (i) slots are allocated to bidders in decreasing order of their bids, so that the highest bids are awarded the slots with the highest prominence values; and (ii) each bidder that is assigned to a slot pays how much their presence in the auction affects the other bidders' social welfare.
%
In truthful mechanisms, allocation and payments are defined so that it is a dominant strategy for each bidder to report their true valuation to the mechanism, namely $b_i =   v_i$ for every $i \in \N$.
In particular, the allocation implemented by the VCG mechanism orderly assigns the first $m$ bidders in decreasing value of $b_i$ to the first $m$ slots (those with the highest click through rates).
At the same time, assuming w.l.o.g. that bidder $i$ is assigned to slot $i$, the mechanism defines an expected payment $p_i \coloneqq \sum_{j=i+1}^{m+1} b_{j}(\lambda_{j-1} -\lambda_{j} )$ for each bidder $i \in \{1, \ldots, m\}$, where, for the ease of notation, we let $\lambda_{m+1} = 0$.
The payment is zero for all the other bidders.
In practice, each bidder $i \in \{1, \ldots, m\}$ has to pay $\frac{p_i}{ \lambda_i}$ whenever a user clicks on their ad, so that their utility is $\lambda_i v_i - p_i$ in expectation over the clicks.
The expected utility of all the other bidders is zero.

\begin{figure}[!htp]
	\centering
	\includegraphics[width=0.8\textwidth]{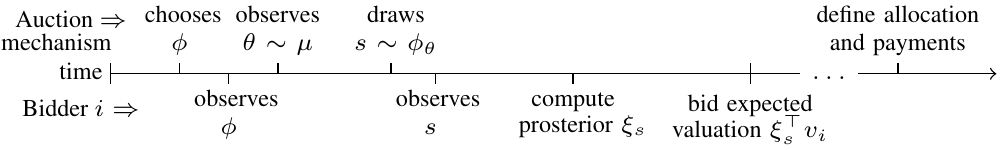}
	\caption{Time-line of a Bayesian ad auction.}
	\label{fig:inter}
\end{figure}

We study \emph{Bayesian} ad auctions, which are characterized by a set $\Theta \coloneqq \{ \theta_1, \ldots,\theta_d \}$ of $d$ states of nature.
Each bidder $i \in \N$ has a valuation vector $ v_i \in [0,1]^d$, with $ v_i(\theta)$ being bidder $i$'s valuation in state $\theta \in \Theta$, and all such vectors are arranged in a matrix of bidders' valuations ${V} \in [0,1]^{n \times d}$, whose entries are defined as ${V}(i,\theta) \coloneqq {v}_i(\theta)$ for all $i \in \mathcal{N}$ and $\theta \in \Theta$.
We model signaling by means of the \emph{Bayesian persuasion} framework~\cite{kamenica2011bayesian}.
We consider the case in which the auction mechanism knows the state of nature and acts as a \emph{sender} by issuing signals to the bidders (the \emph{receivers}), so as to partially disclose information about the state and increase revenue. 
As customary in the literature, we assume that the state is drawn from a common prior distribution $\mu \in \Delta_{\Theta}$, with $\mu_\theta$ denoting the probability of state $\theta \in \Theta$.\footnote{Given a finite set $X$, we denote with $\Delta_X$ the ($|X|-1$)-dimensional simplex defined over the elements of $X$.}
The mechanism publicly \emph{commits to a signaling scheme} $\phi$, which is a randomized mapping from states of nature to signals for the bidders.
We focus on the case of \emph{public signaling} in which all the bidders receive the same signal from the auction mechanism.
Formally, a signaling scheme is a function $\phi:\Theta\to\Delta_{\sset}$, where $\sset $ is a set of available signals.
For the ease of notation, we let $\phi_\theta(s)$ be the probability of sending signal $s\in \sset$ when the state is $\theta \in \Theta$.

A Bayesian ad auction goes on as follows (see Figure~\ref{fig:inter} for a picture): (i) the auction mechanism commits to a signaling scheme $\phi$, and the bidders observe it; (ii) the mechanism gets to know the state of nature $\theta \sim \mu$ and draws signal $s \sim \phi(\theta)$; and (iv) the bidders observe the signal $s$ and rationally update their prior belief over states according to Bayes rule.
After observing signal $s \in \sset$, all the bidders infer a posterior distribution $\xi_{s} \in \Delta_{\Theta}$ over states (also called \emph{posterior} for short) such that the posterior probability of state $\theta \in \Theta$ is
\begin{equation}\label{eq:posterior}
\xi_{s}(\theta) \coloneqq \frac{{\mu}_\theta {\phi}_{\theta}(s) }{\sum_{\theta'\in\Theta}{\mu}_{\theta'} {\phi}_{\theta'}(s) } .
\end{equation}
Finally, each bidder $i \in \N$ truthfully reports to the mechanism their expected valuation given the posterior $\xi_s$, namely $\xi_s^\top v_i  = \sum_{\theta \in \Theta} v_i(\theta) \, \xi_{s}(\theta)$, and the mechanism allocates slots and defines payments as in a standard ad auction.

\paragraph{Representing Signaling Schemes.}
It is oftentimes useful to represent signaling schemes as convex combinations of the posteriors they can induce~\cite{dughmi2014hardness,cheng2015mixture}.
Formally, a signaling scheme $\phi: \Theta \to \Delta_{\sset}$ \emph{induces} a probability distribution $\gamma$ over posteriors in $\Delta_{\Theta}$, with $\gamma(\xi)$ denoting the probability of posterior $\xi \in \Delta_{\Theta}$, defined as
\[
{\gamma}({\xi})  \coloneqq \sum_{s \in \sset:  {\xi}_{s} = \xi } \sum_{\theta \in \Theta} {\mu}_\theta {\phi}_{\theta}(s) .
\]
Indeed, we can directly reason about distributions $\gamma$ over $\Delta_{\Theta}$ rather than about signaling schemes, provided that they are \emph{consistent} with the prior.
By letting $\text{supp}(\gamma) \coloneqq \{\xi \in \Delta_{\Theta} \mid \gamma(\xi) > 0 \}$ be the support of $\gamma$, this requires that
\begin{equation}\label{eq:consistent}
\sum_{\xi \in \text{supp}(\gamma)} \gamma(\xi) \, \xi(\theta)=\mu_\theta \quad \forall \theta \in \Theta.
\end{equation}
In the rest of the paper, we will use the term signaling scheme to refer to a consistent distribution $\gamma$ over $\Delta_{\Theta}$.

\paragraph{Computational Problems.}
We focus on the problem of computing an \emph{optimal} signaling scheme, \emph{i.e.}, one maximizing the revenue of the mechanism.
We study two settings:
\begin{itemize}
	\item the \emph{known valuations} (KV) setting in which the matrix of bidders' valuations $V$ is known to the mechanism; and
	\item the \emph{random valuations} (RV) setting in which the matrix of bidders' valuations $V$ is unknown, but randomly drawn according to a probability distribution $\V$. 
\end{itemize}
As it is customary in the literature (see, \emph{e.g.},~\cite{badanidiyuru2018targeting}), in the RV setting we assume that algorithms have access to a black-box oracle returning i.i.d. samples drawn from $\V$ (rather than actually knowing such distribution).
We denote by $\textsc{Rev}({V},\xi)$ the expected revenue of the mechanism when the bidders' valuations are given by $V$ and the posterior induced by the mechanism is $\xi \in \Delta_{\Theta}$.
Formally, given that bidders truthfully report their expected valuations and assuming w.l.o.g. that bidder $i$ is assigned by the mechanism to slot $i$, we can write $\textsc{Rev}({V},\xi) \coloneqq \sum_{j=1}^{m} j \, \xi^\top v_{j+1}(\lambda_{j} -\lambda_{j+1} )$.
Then, given a signaling scheme $\gamma$, the expected revenue of the mechanism is $\sum_{\xi \in \text{supp}(\gamma)} \gamma(\xi) \, \textsc{Rev}({V},\xi)$.
When the valuations are unknown, we let $\textsc{Rev}(\V,\xi) \coloneqq \mathbb{E}_{V \sim \V} \textsc{Rev}({V},\xi)$ and define the expected revenue analogously.
Notice that, given a distribution of valuations $\V$ (or, in the KV setting, a matrix of bidders' valuations $V$) and a finite set $\Xi \subseteq \Delta_{\Theta}$ of posteriors, it is possible to formulate the problem of computing an optimal signaling scheme as an LP, as follows:\footnote{LP~\ref{eq:lp_signaling} is written for the RV setting, its analogous for the KV setting can be obtained by substituting $\textsc{Rev}(\V,\xi)$ with $\textsc{Rev}(V,\xi)$.}
\begin{subequations}\label{eq:lp_signaling}
	\begin{align}
		\max_{\gamma \in \Delta_{\Xi}} & \,\, \sum_{\xi \in \Xi} \gamma(\xi) \, \textsc{Rev}(\V,\xi) \quad \text{s.t.} \\
		& \sum_{\xi \in \Xi} \gamma(\xi) \, \xi(\theta)=\mu_\theta & \forall \theta \in \Theta. \label{eq:lp_signaling_cons}
	\end{align}
\end{subequations}
In the following, we let $OPT_{\Xi}$ be the optimal value of LP~\ref{eq:lp_signaling}, while we denote with $OPT$ the optimal expected revenue of the mechanism over all the possible signaling schemes $\gamma$.\footnote{The dependence of $OPT_{\Xi}$ and $OPT$ from either $V$ or $\V$ is omitted, as it will be clear from context.}
%
% (in the rest of the paper, the dependence of $OPT_{\Xi}$ and $OPT$ from $V$ or $\V$ is omitted, as it will be clear from context).

%Notation:
%
%\begin{itemize}
%    \item The set of bidders $\mathcal{A}=\{ 1,...,n\}$.
%    \item The set of distinct items $\mathcal{K}=\{ 1,...,k\}$
%    \item The prominence set $\Lambda= \{ \lambda_0, \lambda_1,...,\lambda_n\}$ with $\lambda_0=\lambda_1 \ge \lambda_2 ...\ge \lambda_k \ge 0=\lambda_{k+1}=...=\lambda_{n}$ with $\lambda_j \in [0,1]$ $\forall j \in \mathcal{I} $.
%    \item The payments of bidders given by a VCG mechanism i.e. $p_i^{VCG}=\sum_{j=i }^n b_{j+1}(\lambda_j -\lambda_{j+1})$.
%    \item $V \in [0,1]^{n \times d}$ valuation matrix, $[V]_{i, \theta} := v_i(\theta)$ valuation of bidder $i$ given the state of nature $\theta$
%\end{itemize}

\section{A General Inapproximability Result}

We start our analysis with the following negative result:
\begin{restatable}{theorem}{hardnessOne}\label{thm:hardness_ptas}
 The problem of computing an optimal signaling scheme does \textnormal{not} admit a \textnormal{PTAS} unless $\mathsf{P}=\mathsf{NP}$, even when it is restricted to the KV setting.
\end{restatable}

Theorem~\ref{thm:hardness_ptas} is proved by a reduction from the VERTEX COVER problem in cubic graphs~\cite{APXAlimonti}.
%, by leveraging the fact that it is $\mathsf{NP}$-hard to find a vertex cover whose size is at most $(1+\varepsilon)$ times that of a minimum-sized vertex cover, for a given constant $\varepsilon > 0$~\cite{APXAlimonti}.
%
%, and it uses an even in the easier case in which the auction mechanism knows the matrix of bidders' valuations.

In the rest of this work, we study several settings in which the negative result in Theorem~\ref{thm:hardness_ptas} can be circumvented, by either fixing some parameters of the problem (see Sections~\ref{sec:known valuations}~and~\ref{subsec:unknown_parametrized}) or considering instances with a specific structure (see Sections~\ref{sec:single_minded}~and~\ref{subsec:unknown_non_zero}). 
\section{KV Setting: Parametrized Complexity}\label{sec:known valuations}

%Given the negative result in Theorem~\ref{thm:hardness_ptas}, 
In this section, we study the parametrized complexity of the problem of computing an optimal signaling scheme, showing that it admits a polynomial-time algorithm when either the number of slots $m$ or the number of states of nature $d$ is fixed.

% Before providing our results, we need to introduce some additional notation.
%
In the following, we let $\Pi_l \subseteq 2^{\N}$ be the set of all the the possible permutations of $l \leq n$ bidders taken from $\N$, with $\pi = (i_1,...,i_l) \in \Pi_l$ denoting a tuple made by bidders $i_1, \ldots, i_l \in \N$, in that order. 
We also let $\Xi_{\pi} \subseteq \Delta_{\Theta}$ be the (possibly empty) polytope of posteriors
in which the expected valuations of bidders in $\pi \in \Pi_l$ are ordered (from the highest to the lowest) according to $\pi$; formally, it holds $\Xi_{\pi} \coloneqq \left\{ \xi \in \Delta_\Theta \mid  \xi^\top v_{i_1} \ge \xi^\top v_{i_2} \ge \ldots  \ge \xi^\top v_{i_l} \right\}$.
Notice that, given a permutation $\pi \in \Pi_{l}$ of $l \geq m+1$ bidders, the expected revenue of the mechanism in any posterior $\xi \in \Xi_{\pi}$ is $\textsc{Rev}(V,\xi) = \xi^\top \sum_{j=1}^{m} j v_{i_{j + 1}} (\lambda_{j}-\lambda_{j+1})$, since the bidders truthfully report their expected valuations to the mechanism, and, thus, the latter allocates slots to bidders in $\pi$ according to their order in the permutation.
Thus, for any fixed $\pi \in \Pi_{l}$ with $l \geq m+1$, the term $\textsc{Rev}(V,\xi)$ is linear in $\xi$ over $\Xi_{\pi}$.

\subsection{Fixing the Number of Slots $m$}

In this case, the problem can be solved in polynomial time by formulating it as an LP, thanks to the following lemma:
%
%
%
% FIRST VERSION
%
%In the first setting we are going to present we will consider the number of slots $k$ as a fixed parameter proving that there exists a poly-algorithm for finding a revenue-maximizing signaling scheme.
%%
%To show that we will introduce $\mathcal{P}_l$ as the union of all the possible permutations with $l$ elements of the set of bidders $\mathcal{A}$ and we will associate to each tuple $b=(b_1,...,b_l) \in \mathcal{P}_l$ the set (possibly empty) of posterior probabilities $\Xi_{b} \subseteq \Xi$ satisfying the following relations $\Xi_{b} := \{ \xi \in \Xi \hspace{2mm}|\hspace{2mm} \xi v_{b_1} \ge \xi v_{b_2} \ge ...  \ge \xi v_{b_l} \}$.
%Note that as long as $\xi \in \Xi_b$ the $i$-th highest expected bid will coincide with the one of the receiver $b_i$ and so, if the tuple $b$ has at least $k+1$ components, the expected revenue of the auctioneer will be given by: 
%$Rev(\xi, V)= \sum_{\theta \in \Theta} \xi(\theta) \hspace{0.5mm}  (\hspace{0.5mm}  \sum_{j=1}^{k} \hspace{0.5mm}  \hspace{0.5mm}(\lambda_{j+1}-\lambda_j)\hspace{0.5mm} j  \hspace{0.5mm} v_{b_{j+1}}(\theta) )$
%which clearly results a liner function in $\xi \in \Xi_b$ for each $b \in \mathcal{P}_l$ as long as $l \ge k+1$. Thanks to the latter observation the following lemma holds.
%
\begin{restatable}{lemma}{lemsupport}\label{lem:support}
	There always exists an optimal signaling scheme $\gamma$ such that $|\Xi_\pi \cap \textnormal{supp}(\gamma)| \leq 1$ for every $\pi \in \Pi_{m+1}$.
\end{restatable}
%
% FIRST VERSION
%
%\begin{restatable}{lemma}{lemsupport}\label{lem:support}
%	Let $\pi \in \Pi_{k+1}$ and $\gamma$ be an optimal distribution over $\Delta_\Theta$ with $\xi_1,\xi_2 \in \Xi_{b}$ elements of its support then there exists an optimal signaling scheme $\gamma^* \in \Delta_\Xi$ with a smaller support size
%\end{restatable}
%
%
Intuitively, the lemma follows from the fact that, given any signaling scheme $\gamma$ and two posteriors $\xi,\xi' \in \text{supp}(\gamma)$ such that $\xi,\xi' \in \Xi_\pi$ for some $\pi \in \Pi_{m+1}$, it is always possible to define a new signaling scheme that replaces $\xi$ and $\xi'$ with a suitably-defined convex combination of them, without decreasing the expected revenue (since it is linear over $\Xi_\pi$).

%By Lemma~\ref{lem:support}, we can write the revenue maximization problem as $\max_{\gamma, \xi_\pi} \sum_{\pi \in \Pi_{m+1}} \gamma(\xi_\pi) \textsc{Rev} (V, \xi_\pi)$ subject to constraints ensuring that each $\xi_\pi $ belongs to $\Pi_{m+1}$ and that $\gamma$ is a consistent distribution over posteriors $\xi_\pi$ (see Equation~\eqref{eq:consistent}).
By Lemma~\ref{lem:support}, we can re-write the revenue maximization problem as $\max \sum_{\pi \in \Pi_{m+1}} \gamma(\xi_\pi) \textsc{Rev} (V, \xi_\pi)$ subject to constraints ensuring that each $\xi_\pi $ belongs to $\Xi_\pi$ (for $\pi \in \Pi_{m+1}$) and that $\gamma$ is a consistent probability distribution over such posteriors (see Equation~\eqref{eq:consistent}).
This problem can be formulated as an LP by introducing a variable for each $\pi \in \Pi_{m+1}$ and $\theta \in \Theta$, encoding the products $\gamma (\xi_\pi) \xi_\pi(\theta)$ that define the expected revenue.
Overall, the resulting LP (see LP~\ref{eqn:LP_kfixed} in the Supplementary Material) has a number of variables and constraints that is $O(n^{m})$, which, after fixing $m$, is polynomial in the size of the input.\footnote{We remark that LP~\ref{eqn:LP_kfixed} in the Supplementary Material is a generalization of the LP presented by~\citeauthor{emek2014signaling}~[\citeyear{emek2014signaling}] for the easier case of second price auctions.}
Thus, we conclude that:
%
%
% FIRST VERSION
%
%Finally, to find the optimal value of a revenue maximizing signaling scheme we will design a LP taking as variables the following quantity $\xi_{b}(\theta) := \gamma(\xi)\hspace{0.2mm} \xi(\theta)$  $ \forall \theta \in \Theta$ and for each $b \in \mathcal{P}_{k+1}$ which simiply represents the posterior probability scaled by the probability to be drawn. So thtat, by lemma (\ref{lem:support}), there will be at most one single $\xi \in \Xi_{b}$ in the support of an optimal signaling scheme and thus one single $\xi_b$ different from zero for each $b \in \mathcal{P}_{k+1}$
%Sinche the total number of variables is equal to $|\mathcal{P}_{k+1}|=O(n^{k+1})$ variables as long as $k$ is fixed it is possibly to define an LP solvable in polynomial time.
%
\begin{restatable}{theorem}{lemmaKnownD}\label{lem:lemmaKnownD}
	In the KV setting, if the number of slots $m$ is fixed, then an optimal signaling scheme can be computed in polynomial time.
\end{restatable}

\subsection{Fixing the Number of States $d$}

Our polynomial-time algorithm exploits the fact that an optimal signaling scheme can be computed by restricting the attention to distributions supported on a finite set of posteriors whose cardinality is polynomial in all the parameters, except from $d$.
In particular, it is sufficient to focus on the set $\Xi^* \coloneqq \bigcup_{\pi \in \Pi_{n}} V(\Xi_{\pi})$, where $V(\cdot)$ denotes the set of vertices of the polytope given as input.
Formally:
%
%
% FIRST VERSION	
%
%We will state an analogous result when the size of the set of states of nature has a fixed size $d$ and the auctioneer can observe the receivers' valuation matrix.
%Similarly to what we have done so far we will associate to each tuple made by $n$ bidders $b \in \mathcal{P}_{n}$ the sub-regions (possibly empty) of the simplex $\Xi_{b} \subseteq \Xi$ defined before.
%Moreover we will indicate with $V(\Xi_{b})$ the set of vertexes of $\Xi_{b}$ while we will define $\Xi^*$ as the union of the these sets over the different tuples $b \in \mathcal{P}_n$, formally we will have that $\Xi^* := \bigcup_{b \in \mathcal{P}_{n}} V(\Xi_{b})$.
%The latter set will play a crucial in this scenario since, in the following lemma, we will show that there exists an optimal signaling scheme whose support is a subset of $\Xi^*$.
%
\begin{restatable}{lemma}{lemmaKnownfixd}\label{lem:knownfixD}
	It holds that $OPT_{\Xi^*} = OPT$.
	%
	%There always exists an optimal signaling scheme $\gamma$ such that $\textnormal{supp}(\gamma) \subseteq \Xi^*$.
\end{restatable}
The lemma follows from the fact that, given any signaling scheme $\gamma$ and posterior $\xi \in \text{supp}(\gamma)$ such that $\xi \in \Xi_\pi$ for some $\pi \in \Pi_n$, by Carathéodory's theorem it is always possible (since $\Xi_\pi $ is a polytope) to decompose $\xi$ into a convex combination of the vertices of $\Xi_{\pi}$, obtaining a new signaling scheme that provides the mechanism with an expected revenue at least as large as that of $\gamma$ (since $\textsc{Rev}( V, \xi)$ is linear over $\Xi_\pi$).
By observing that $|\Xi^*|= O( (n^2+d)^{d-1})$, it is easy to show that an optimal signaling scheme can be computed by means of LP~\ref{eq:lp_signaling} instantiated for the set $\Xi^*$, which has a number of variables and constraints that is polynomial once $d$ is fixed.
This proves the following:
%
%
% FIRST VERSION
%
%The lemma is a consequence of the fact that in all the regions $\Xi_b$ for each $b \in \mathcal{P}_n$ the revenue results a linear function so that it is always possible to decompose each posterior probability $\xi \in \Xi_{b}$ by means of Caratheodory's Theorem as a convex combination of the vertexes of $\Xi_{b}$ without decreasing the final revenue. Finally observing that $O(|\Xi^*|= (n+d)^{d-1})$, as long as d is fixed, we can design LP for finding a revenue-maximizing signaling scheme in polynamial time as sated in the nex theorem.
%
\begin{restatable}{theorem}{thmknownfixd}\label{thm:knownfixD}
	In the KV setting, if the number of states $d$ is fixed, then an optimal signaling scheme can be computed in polynomial time.
\end{restatable}
\section{KV Setting: Single-Minded Bidders}\label{sec:single_minded}

In this section, we focus on particular Bayesian ad auctions where the bidders are \emph{single minded}.
Intuitively, in our setting, by single mindedness we mean that each bidder is interested in displaying their ad only when the realized state of nature is a specific (single) state, and that all the bidders interested in the same state value a click on their ad for the same amount.
We introduce the following formal definition:
\begin{definition}[Single-minded bidders]\label{def:single_minded}
	In a Bayesian ad auction, we say that bidders are \emph{single minded} if there exist $\N_\theta \subseteq \N$ and $\delta_\theta \in [0,1]$ for all $\theta \in \Theta$ such that:
	\begin{itemize}
		\item[(i)] $\N = \bigcup_{\theta \in \Theta} \N_\theta$ and $\N_\theta \cap \N_{\theta'} = \emptyset$ for all $\theta \neq \theta' \in \Theta$;
		\item[(ii)] for every $\theta \in \Theta$ and $i \in \N_\theta$, it holds $v_i(\theta) = \delta_\theta $ and $v_i(\theta') = 0$ for all $\theta' \in \Theta: \theta' \neq \theta $.
	\end{itemize}
\end{definition}

Notice that, given a posterior $\xi \in \Delta_{\Theta}$ induced by the mechanism, all the bidders $i$ belonging to the same set $\N_\theta$ have the same expected valuation, namely $\xi^\top v_i = \delta_\theta \xi(\theta)$ for all $\theta \in \Theta$ and $i \in \N_\theta$.
As a result, given that bidders truthfully report their expected valuations, the mechanism will always receive at most $d$ different bids, one per set $\N_\theta$.

The last observation implies that, given $\xi \in \Delta_{\Theta}$, in order to unequivocally define an allocation of bidders to slots (and, thus, also define the expected payments) it is sufficient to know the relative ordering of the (at most) $d$ different expected valuations associated to sets $\N_\theta$.
This allows us to tackle the problem with an approach analogous to the one of Section~\ref{sec:known valuations}, with the only difference that, in this case, we will reason about permutations of the groups of bidders $\N_\theta$, rather than about permutations of all the individual bidders.

In the following, we let $\Pi \subseteq 2^\Theta$ be the set of all the permutations of the sates of nature $\Theta = \{ \theta_1, \ldots, \theta_d \}$, while we let $\pi = (\theta_{k_1}, \ldots, \theta_{k_d}) \in \Pi$ be an ordered tuple made by states $\theta_{k_1}, \ldots, \theta_{k_d} \in \Theta$, where $k_1,\ldots,k_d \in \{1,\ldots,d\}$.
Moreover, $ \Xi_{\pi} \coloneqq \left\{ \xi \in \Delta_{\Theta} \mid \delta_{\theta_{k_1}} \xi(\theta_{k_1}) \geq  \ldots \geq \delta_{\theta_{k_d}} \xi(\theta_{k_d}) \right\} $ is the polytope of posteriors in which the expected valuations associated to sets $\N_\theta$ are ordered according to $\pi$.
%
% Let us also notice the term $\textsc{Rev}(V,\xi)$ is linear in $\xi$ over $\Xi_{\pi}$.

The first preliminary result that we need in order to derive our approximation algorithm is a characterization of the vertices of the sets $\Xi_\pi$ for $\pi \in \Pi$, as follows.
%
% This will allow us to formulate the problem of computing an optimal signaling scheme as an LP with exponentially-many variables and polynomially-many constraints, whose dual can be solved in polynomial time by means of the ellipsoid algorithm.
%
\begin{restatable}{lemma}{uniformVertices}\label{lem:uniformVertices}
	Given $\pi \in \Pi$ and $\xi \in \Xi_\pi$, it holds that $\xi \in V(\Xi_\pi)$ if and only if there exists $\ell \in \{1,\ldots, d\}$ such that:
	\begin{itemize}
		\item[(i)] $\delta_{\theta_{k_1}} \xi(\theta_{k_1}) = \ldots = \delta_{\theta_{k_\ell}} \xi(\theta_{k_\ell}) > 0$; and 
		\item[(ii)] $\delta_{\theta_{k_{\ell+1}}} \xi(\theta_{k_{\ell+1}}) = \ldots = \delta_{\theta_{k_d}} \xi(\theta_{k_d}) = 0$.
	\end{itemize}
\end{restatable}
Intuitively, Lemma~\ref{lem:uniformVertices} states that the vertices of a set $\Xi_\pi$ are all the posteriors $\xi \in \Delta_\Theta$ such that, for some $\ell \in \{1, \ldots, d\}$, only the first $\ell$ states according to the ordering defined by $\pi$ are assigned a positive probability, while all the remaining states have zero probability.
Moreover, the positive probabilities of the posterior $\xi$ are defined so that all the bidders belonging to the first $\ell$ sets $\N_\theta$, according to the ordering defined by $\pi$, are the same.
Notice that, in the special case in which all the values $\delta_\theta$ are equal to one, the vertices of all the sets $\Xi_\pi$ are all the uniform probability distributions over subsets of $\ell$ states of nature, for any $\ell \in \{1, \ldots, d\}$.

%
%by letting $\Xi^* = \bigcup_{\pi \in \Pi} V(\Xi_\pi)$, we can prove the following:
%%
%\begin{lemma}\label{lem:uniform_vertices}
%	$\Xi^* = \Delta_\Theta \cap  \left( \bigcup_{\ell = 1}^d \left\{ 0, \frac{1}{\ell} \right\}^d \right)$.
%\end{lemma}
%%
%
%Intuitively, Lemma~\ref{lem:uniform_vertices} states that the set $\Xi^*$ is made by all the posteriors defined as uniform probability distributions over subsets of $\ell$ states of nature, for any $\ell \in \{1, \ldots, d\}$.
%%
%Its proof follows from the crucial observation that, by some linear programming arguments, for any vertex $\xi \in V(\Xi_\pi)$ with $\pi = (\theta_{k_1}, \ldots, \theta_{k_d}) \in \Pi$, there exists $\ell \in \{1, \ldots, d\}$ such that $\delta_{\theta_{k_1}} \xi(\theta_{k_1}) = \ldots = \delta_{\theta_{k_\ell}} \xi(\theta_{k_\ell}) > 0$ and, additionally, $\delta_{\theta_{k_{\ell+1}}} \xi(\theta_{k_{\ell+1}}) = \ldots = \delta_{\theta_{k_d}} \xi(\theta_{k_d}) = 0$.

By letting $\Xi^* = \bigcup_{\pi \in \Pi} V(\Xi_\pi)$, since the term $\textsc{Rev}(V,\xi)$ is linear in $\xi$ over $\Xi_{\pi}$ for every permutation $\pi \in \Pi$, we can conclude that $OPT_{\Xi^*} = OPT$ (the proof is analogous to that of Lemma~\ref{lem:knownfixD}).
Thus, Lemma~\ref{lem:uniformVertices} allows us to find an optimal signaling scheme by solving LP~\ref{eq:lp_signaling} for the set $\Xi^*$ and the matrix of bidders' valuations $V$.
However, notice that, since the size of $\Xi^*$ is exponential in $d$, the resulting LP has exponentially-many variables.
Nevertheless, since the LP has polynomially-many constraints, we can still solve it in polynomial time by applying the ellipsoid algorithm to its dual, provided that a polynomial-time separation oracle is available.

In order to design a polynomial-time separation oracle, we apply the procedure described above to a relaxed version of LP~\ref{eq:lp_signaling}, whose optimal value is sufficiently ``close'' to that of the original LP.
Given $\beta \in \mathbb{R}_+$, the relaxed LP reads as follows:
\begin{subequations}\label{eq:lp_signaling_relax_main}
	\begin{align}
	\max_{\gamma \in \Delta_{\Xi^*},z\le 0} & \,\, \sum_{\xi \in \Xi} \gamma(\xi) \, \textsc{Rev}(\V,\xi) +   \beta z \quad \text{s.t.} \\
	& \sum_{\xi \in \Xi^*} \gamma(\xi) \, \xi(\theta) -  z  \ge \mu_\theta & \forall \theta \in \Theta. 
	\end{align}
\end{subequations}
The dual problem of LP~\ref{eq:lp_signaling_relax_main} reads as follows:
\begin{subequations}\label{eq:dual_lp_signaling}
	\begin{align}
	\min_{y \leq 0, t} & \,\, \sum_{\theta \in \Theta} y_\theta \mu_\theta +t  \quad \text{s.t.} \\
	& \sum_{\theta \in \Theta} y_\theta \, \xi(\theta) + t \geq \textsc{Rev}(\V,\xi) & \forall \xi \in \Xi^* \\
	& \sum_{\theta \in \Theta} y_\theta \ge -\beta, \label{eq:constraint_relaxed}
	\end{align}
\end{subequations}
where $y_\theta$ for $\theta \in \Theta$ are dual variables associated to Constraints~\eqref{eq:lp_signaling_cons}, while $t$ is a dual variable for $\sum_{\xi \in \Xi^*} \gamma(\xi) = 1$.
Notice that, by relaxing the LP, in the dual LP~\ref{eq:dual_lp_signaling} we get the additional Constraint~\eqref{eq:constraint_relaxed} and that $y_\theta \leq 0$ for all $\theta \in \Theta$.
This is crucial to design a polynomial-time separation oracle.

The separation problem associated to Problem~\ref{eq:dual_lp_signaling} reads as:
%
% In order to make the ellipsoid algorithm running on LP~\ref{eq:dual_lp_signaling} in polynomial time, it is necessary to implement a suitable polynomial-time algorithm for the separation problem associated to LP~\ref{eq:dual_lp_signaling}, which reads as follows:
%
%\footnote{In particular, we apply the ellipsoid to the dual of a relaxation of LP \ref{eq:dual_lp_signaling} parametrized by a value $\beta$. In this way, the dual has the additional constraint $y_\theta \in[-\beta,0]$.}
%
\begin{definition}[Separation problem]\label{def:separation_problem}
	Given values for the dual variables $y_{\theta} \in [-\beta,0]$ for all $\theta \in \Theta$, compute:
	\begin{equation}\label{eq:separation}
	\max_{\xi \in \Xi^*} \,\, \textsc{Rev}(V, \xi) - \sum_{\theta \in \Theta}    y_{\theta} \, \xi(\theta).
	\end{equation}
\end{definition}
The following Lemma~\ref{lem:dp} shows that Problem~\ref{eq:separation} can be solved optimally up to any given additive loss $\lambda > 0$, by means of a \emph{dynamic programming} algorithm that runs in time polynomial in the size of the input, in $\frac{1}{\lambda}$, and in $\beta$.
Formally:
\begin{restatable}{lemma}{dinprog}\label{lem:dp}
	Given $\lambda > 0$, there exists an algorithm that finds an additive $\lambda$-approximation to Problem~\ref{eq:separation}, in time polynomial in the size of the input, in $\frac{1}{\lambda}$, and in $\beta$.
\end{restatable}
The crucial observation that allows us to solve Problem~\ref{eq:separation} by means of dynamic programming is that, in any posterior $\xi \in \Xi^*$, bidders' expected valuations are either a positive, bidder-independent value or zero (see Lemma~\ref{lem:uniformVertices}).
This allows us to build a discretized range of possible bidders' valuation values, so that, for each discretized value, we can compute an optimal posterior $\xi \in \Xi^*$ inducing that value by adding states of nature incrementally in a dynamic programming fashion.

% \textcolor{red}{The crucial observation which allows us to design an FPTAS by means of dynamic programming for solving ~\ref{lem:dp} is that given $\xi \in \Xi^*$ the revenue is a function of the subset of states of nature which belongs to the support of $\xi$.} 

Since the algorithm in Lemma~\ref{lem:dp} only returns an approximate solution to Problem~\ref{eq:separation}, we need to carefully apply the ellipsoid algorithm to solve LP~\ref{eq:dual_lp_signaling}, so that it correctly works even with an approximated oracle.
Some non-trivial duality arguments allow us to prove that, indeed, this can be achieved by only incurring in a small additive loss on the quality of the returned solution, and without degrading the running time of the algorithm.
Overall, this allows us to conclude that:
\begin{restatable}{theorem}{fptasSingleMinded}\label{thm:fptasSingleMinded}
	In the KV setting, if the bidders are single minded, then the problem of computing an optimal signaling scheme admits an (additive) FPTAS. 
\end{restatable}

\section{RV Setting}\label{sec:unknowm_valuations}

In this setting, as stated in Section~\ref{sec:preliminaries}, we assume that the auction mechanism has access to the distribution of bidders' valuations $\V$ only through a black-box sampling oracle.
In the following, given $s \in \mathbb{N}_{>0}$ i.i.d. samples of matrices of bidders' valuations, namely $V_1, \ldots, V_s \in [0,1]^{n\times d}$, we let $\V^s$ be their empirical distribution, which is such that $\mathbb{P}_{V \sim \V^s} \left\{  V = \hat V \right\} \coloneqq \frac{\sum_{t = 1}^s \mathbbm{1} \{ V_t = \hat V \} }{s}$ for all $\hat V \in [0,1]^{n \times d}$.

In this section, we first study the parametrized complexity of the problem of computing an optimal signaling scheme in general auctions (Section~\ref{subsec:unknown_parametrized}), and, then, we address special auction settings in which the bidders' valuations are \emph{bounded away from zero}, namely $v_i(\theta) > \delta$ for all $i \in\ N$ and $\theta \in \Theta$, for some threshold $\delta > 0$.
In the latter case, we show that the problem admits a QPTAS and the result is tight (Section~\ref{subsec:unknown_non_zero}).

Before stating our main results (Theorems~\ref{lem:UnkownFixedDThm},~\ref{lem:UnkownFixedKThm},~\ref{lem:UnkownFixedKHigh},~and~\ref{lem:UnkownFixedKHighTight}), we introduce some preliminary useful lemmas.
The first one (Lemma~\ref{lem:rvFinite}) works under the true distribution of bidders' valuations $\V$, and it establishes a connection between the optimal expected revenue ($OPT$) and the optimal value of LP~\ref{eq:lp_signaling} for suitably-defined finite sets $\Xi \subseteq \Delta_{\Theta}$ of posteriors ($OPT_{\Xi}$).
In particular, we look at sets $\Xi \subseteq \Delta_{\Theta}$ for which the function $\textsc{Rev}(\V, \cdot)$ is ``\emph{stable}'' according to the following definition:\footnote{Notions of stability analogous to that in Definition~\ref{def:stability} have already been used in the literature; see, \emph{e.g.},~\cite{cheng2015mixture}.}
\begin{definition}[$(\alpha, \varepsilon)$-stability]\label{def:stability}
	Given $\alpha, \varepsilon \geq 0$ and a finite set $\Xi \subseteq \Delta_{\Theta}$, we say that $\textsc{Rev}(\V, \cdot)$ is \emph{$(\alpha, \varepsilon)$-stable} for $\Xi$ if, for every $ \xi \in \Delta_{\Theta}$, there exists a distribution $\gamma_\xi \in \Delta_{\Xi}$ such that:
	\begin{equation}\label{eq:prop_finite_set}
	\sum_{\xi' \in \Xi} \gamma_\xi (\xi') \textsc{Rev}(\mathcal{V}, \xi') \ge (1-\alpha)\textsc{Rev}(\mathcal{V}, \xi) - \varepsilon.
	\end{equation}
\end{definition}
%
% given $\alpha, \varepsilon > 0$, we look at finite sets $\Xi \subseteq \Delta_{\Theta}$ such that, for every posterior $ \xi \in \Delta_{\Theta}$, there exists a probability distribution $\gamma_\xi \in \Delta_{\Xi}$ (supported on $\Xi$) that satisfies:
%\begin{equation}\label{eq:prop_finite_set}
%	\sum_{\xi' \in \text{supp}(\gamma_\xi)} \gamma_\xi (\xi') \textsc{Rev}(\mathcal{V}, \xi') \ge (1-\alpha)\textsc{Rev}(\mathcal{V}, \xi) - \varepsilon.
%\end{equation}
%
For any finite set $\Xi \subseteq \Delta_{\Theta}$ such that $\textsc{Rev}(\V, \cdot)$ is \emph{$(\alpha, \varepsilon)$-stable} for $\Xi$, starting from an optimal signaling scheme $\gamma$ one can recover an optimal solution to LP~\ref{eq:lp_signaling}, only incurring in ``small'' multiplicative and additive losses in the expected revenue, respectively of $1-\alpha$ and $\varepsilon$.
This can be accomplished by decomposing each posterior $\xi \in \text{supp}(\gamma)$ into $\gamma_\xi \in \Delta_{\Xi}$ and, then, putting such distributions together.
These observations allow us to prove the following lemma:
\begin{restatable}{lemma}{rvFinite}\label{lem:rvFinite}
	Given $\alpha, \varepsilon \geq 0$ and $\Xi \subseteq \Delta_{\Theta}$ such that $\textsc{Rev}(\V, \cdot)$ is $(\alpha, \varepsilon)$-stable for $\Xi$, it holds $OPT_{\Xi} \geq (1-\alpha)OPT-\varepsilon$.
\end{restatable}

The second lemma (Lemma~\ref{lem:rvSamples}) deals with the approximation error introduced by using an empirical distribution of bidders' valuations $\V^s$, rather than the actual distribution $\V$.
Given a finite set $\Xi \subseteq \Delta_\Theta$ of posteriors, let $\gamma_{\V^s} \in \Delta_{\Xi}$ be an optimal solution to LP~\ref{eq:lp_signaling} for distribution $\V^s$ and set $\Xi$.
Moreover, let $OPT_{\Xi,s} \coloneqq \mathbb{E} \left[ \sum_{\xi \in \Xi} \gamma_{\V^s} (\xi) \textsc{Rev}(\V, \xi) \right]$ be the average expected revenue of signaling schemes $\gamma_{\V^s}$ under the true distribution of valuations $\V$, where the expectation is with respect to the sampling procedure that determines $\V^s$.
Then, a concentration argument proves the following:
\begin{restatable}{lemma}{rvSamples}\label{lem:rvSamples}
	Given $\rho, \tau > 0$, let $\Xi \subseteq \Delta_{\Theta}$ be finite and $s \coloneqq \left\lceil \frac{2(\lambda_1 m)^2}{ \tau^2}\log \frac{2}{\rho} \right\rceil$, $OPT_{\Xi,s} \geq \left( 1-\rho |\Xi| \right) OPT_{\Xi} -\tau$.
	%
	% then $OPT^*$ has value at least equal to $(1-\alpha|\Xi^*|) OPT1 - \varepsilon$ where OPT1 is an optmimal solution of LP(\ref{known_val})
\end{restatable}

Finally, the last lemma (Lemma~\ref{lem:UnkownFixedKLem}) exploits Lemma~\ref{lem:rvFinite} to provide two useful bounds on the value of $OPT_{ \Xi_q}$, where $\Xi_q \subseteq \Delta_{\Theta}$ (for a given $q \in \mathbb{N}_{>0}$) is the finite set of all the $q$-uniform posteriors, according to the following definition:
\begin{definition}[$q$-uniform posterior]
	Given $q \in \mathbb{N}_{>0}$, a posterior $\xi \in \Delta_{\Theta}$ is \emph{$q$-uniform} if each $\xi(\theta)$ is a multiple of $\frac{1}{q}$.
\end{definition}
Notice that the set $\Xi_q$ has size $|\Xi^q| \le \min \{ d^q, q^d \}$.
The two points in the following lemma are readily proved by applying Lemma~\ref{lem:rvFinite}, after noticing that the sets $\Xi_q$ in the statement are such that the function $\textsc{Rev}(\V, \cdot)$ is {$(\alpha, \varepsilon)$-stable} for them, with suitable values of $\alpha \geq 0$ and $\varepsilon \geq 0$.
Formally:
\begin{restatable}{lemma}{UnkownFixedKLem}\label{lem:UnkownFixedKLem}
	Given $\eta > 0$ and $q\coloneqq \left\lceil \frac{1}{2 \eta^2}\log \frac{m+1}{\eta } \right\rceil$, it holds:
	\begin{enumerate}
		\item[(i)] $OPT_{\Xi_q} \geq OPT- 2\eta m$;
		\item[(ii)] if, for some $\delta > 0$, it is the case that $v_i(\theta) > \delta$ for all $i \in \N$ and $\theta \in \Theta$, then $OPT_{\Xi_q} \geq ( 1- \frac{\eta}{\delta} )^2 \, OPT$.
	\end{enumerate}
	%
	%
	%	Let $\delta, \varepsilon>0$ with $v_i(\theta) \ge \delta$ $\forall \theta \in \Theta, \hspace{0.5mm} \forall i \in \mathcal{A}$ then an optimal solution of (3) has value at least :
	%	\begin{enumerate}
	%		\item[(i)] $OPT-2\varepsilon k$
	%		\item[(ii)] $\big( 1- \varepsilon / \delta \big)^2 OPT$
	%	\end{enumerate}
\end{restatable}
%
% \textcolor{red}{the lemma follows from that, as long as $q$ is sufficiently large, we can safely decompose each posterior probability $\xi \in \Delta_{\Theta}$ as a convex combination of the elements of $\Xi_q$ incurring only in a little loss in terms of the expected revenue prescribed by the signaling scheme}

\subsection{Parametrized Complexity}\label{subsec:unknown_parametrized}

First, we study the computational complexity of the problem of computing an optimal signaling scheme when the number of states $d$ is fixed.
We provide an (additive) FPTAS that works by performing the following two steps: (i) it collects a suitable number $s \in \mathbb{N}_{>0}$ of matrices of bidders' valuations, by invoking the sampling oracle; and (ii) it solves LP~\ref{eq:lp_signaling} for the resulting empirical distribution $\V^s$ and a suitably-defined set of $q$-uniform posteriors.
In particular, given a desired (additive) error $\lambda >0$, the algorithm works on the set $\Xi_q$ for $q =  \lceil \frac{md}{\lambda}  \rceil$ and its approximation guarantees rely on the following Lemma~\ref{lem:UnkownFixedDLem}, proved again by means of Lemma~\ref{lem:rvFinite}.
\begin{restatable}{lemma}{UnkownFixedDLem}\label{lem:UnkownFixedDLem}
	For $\lambda>0$ and $q =  \lceil \frac{md}{\lambda}  \rceil$, $OPT_{\Xi_q} \hspace{-1mm} \geq \hspace{-1mm} OPT -\lambda$. 
\end{restatable}
%
%
% FIRST VERSION
%
%Since, as previously discussed, there is no PTAS for finding a revenue-maximizing signaling scheme in the general case we will consider some particular cases. The first one we will present will be the one in which the size of the set of states of nature $d$ is a fixed parameter. As a first step we will introduce the following set $\Xi_d= \{ \xi \in \Xi | \xi(\theta) \in \{ \lambda i / kd   \} \hspace{1mm} i \in [ kd/ \lambda ], \hspace{1mm} \forall \theta \in \Theta \}$ and we will show that an optimal solution of a LP restricted to $\Xi_d$ will provide an expected revenue which will be close to the actual one up to a fixed tolerance.
%
%\begin{restatable}{lemma}{UnkownFixedDLem}\label{lem:UnkownFixedDLem}
%Given $\lambda>0$ setting $\Xi^* = \Xi_d$  in LP(\ref{known_val}) then its optimal solution has value at least $OPT - \lambda $  
%\end{restatable}

Thanks to Lemmas~\ref{lem:rvSamples}~and~\ref{lem:UnkownFixedDLem} (the former applied for suitable values $\rho, \tau >0$), we can prove that the procedure described in steps (i) and (ii) above gives a signaling scheme achieving an expected revenue at most a function of $\lambda$ lower than $OPT$, provided that the number of samples $s$ is defined as in Lemma~\ref{lem:UnkownFixedDLem}.
Moreover, let us notice that, since $|\Xi_q| = O(q^d)= O ( (\frac{1}{\lambda}md)^d)$, if $d$ is fixed, then the overall procedure runs in time polynomial in the input size and in $\frac{1}{\lambda}$.
%\textcolor{red}{Note that in this case we chose the first bound on the cardinality of the set $\Xi_q$ since we have that d is fixed}
%
Thus, we can conclude that:
\begin{restatable}{theorem}{UnkownFixedDThm}\label{lem:UnkownFixedDThm}
	In the RV setting, if the number of states $d$ is fixed, then the problem of computing an optimal signaling scheme admits and (additive) FPTAS.
	%
	% Let d a fixed parameter then there exists an FPTAS for computing a revenue-maximizing signaling scheme
\end{restatable}

Next, we switch the attention to the case in which the number of slots $m$ is fixed.
We provide an (additive) PTAS that works as the FPTAS in Theorem~\ref{lem:UnkownFixedDThm}, but whose approximation guarantees follow from Lemma~\ref{lem:rvSamples} and point (i) in Lemma~\ref{lem:UnkownFixedKLem} (rather than Lemma~\ref{lem:UnkownFixedDLem}).
Thus, the only difference with respect to the previous case is that the algorithm works on the set $\Xi_q$ of $q$-uniform posteriors for $q $ defined as in Lemma~\ref{lem:UnkownFixedKLem}.
As a result, since $|\Xi_q| = O(d^q)$ and $q$ depends on a parameter $\eta > 0$ that is related to the quality of the obtained approximation, the algorithm is only a PTAS rather than an FPTAS.
Formally, we can prove the following:
\begin{restatable}{theorem}{UnkownFixedKThm}\label{lem:UnkownFixedKThm}
	In the RV setting, if the number of slots $m$ is fixed, then the problem of computing an optimal signaling scheme admits and (additive) PTAS.
	%
	%Let k be a fixed parameter then there exists a PTAS for computing an optimal revenue-maximizing signaling scheme
\end{restatable}

\subsection{Valuations Bounded Away From Zero}\label{subsec:unknown_non_zero}

We conclude the section by studying the case in which the bidders' valuations are bounded away from zero.
This case is dealt with an algorithm identical to the one in Theorem~\ref{lem:UnkownFixedKThm}, but carrying on the approximation analysis by using Lemma~\ref{lem:rvSamples} and point (ii) in Lemma~\ref{lem:UnkownFixedKLem} (rater than point (i)).
Thus, since the value of $q$ in Lemma~\ref{lem:UnkownFixedKLem} is related to the quality of the approximation thorough a parameter $\eta > 0$ and also depends logarithmically on the number of slots $m$, we obtain:
\begin{restatable}{theorem}{UnkownFixedKHigh}\label{lem:UnkownFixedKHigh}
In the RV setting, if $v_i(\theta) \ge \delta$ for all $i \in \N$ and $\theta \in \Theta$ for some $\delta > 0$, then the problem of computing an optimal signaling scheme admits a (multiplicative) QPTAS.
\end{restatable}

The following theorem shows that the result is tight.

\begin{restatable}{theorem}{UnkownFixedKHighTight}\label{lem:UnkownFixedKHighTight}
	Assuming the ETH, there exists a constant $\omega > 0$ such that finding a signaling scheme that provides an expected revenue at least of $(1-\omega) OPT$ requires $I^{\tilde \Omega (\log I )}$ time, where $I$ is the size of the problem instance. This holds even when $v_i(\theta) > \frac{1}{3}$ for all $i \in \N$ and $\theta \in \Theta$.\footnote{The $\tilde \Omega$ notation hides poly-logarithmic factors.}
	%
	%Assuming ETH, there exists a constant $\delta>0$ such that finding a $(1-\delta)$ approximation to the signaling problem with utilities larger than $\frac{1}{3}$ requires $n^{\tilde \Omega (log(n))}$ time.\footnote{$\tilde \Omega$ hides polylogarithmic factors.}
\end{restatable}

\clearpage
\bibliographystyle{named}
\bibliography{bibliography}

\clearpage
\appendix
\onecolumn
\begin{center}
	\LARGE{\textbf{Supplementary Material }}
\end{center}

\hardnessOne*
\begin{proof}
We reduce from vertex cover in cubic graphs.
Formally, it is NP-Hard to approximate the minimum size vertex cover in cubic graph with an approximation $(1+\epsilon)$, for a given constant $\epsilon>0$~\citeauthor{APXAlimonti}~[\citeyear{APXAlimonti}].
Let $\eta= \epsilon/7$ and $\delta=\eta/4$.
We show that for $\delta$, an $1-\delta$ approximation to the signaling problem can be used to provide a $(1+\epsilon)$ approximation to vertex cover in polynomial time.

Given an instance of vertex cover $(L,E)$ with nodes $\rho=|L|$ and edges $E$.
For each $z \in \{1,\dots,\rho\}$, we build an instance as follows.
There are $m_z=z+\rho|E|-1$ slots and $\lambda_j=1$ for each $j \in \{1,\dots,m\}$.
The set of states is $\Theta=\{\theta_l\}_{l \in L}$ and the set of receivers is $\N=\{r_{e,i}\}_{e \in E, i \in \{1,\dots,\rho\}} \cup \{r_l\}_{l \in L} \cup \{r_i\}_{i \in \{1,\dots,m+1\}}$.
The valuation of a receiver $r_{e,i}$, $e \in E$ and $i \in \{1,\dots,\rho\}$, is $v_{r_{e,i}}(\theta_v)=1$ if $v \in e$, \emph{i.e.}, $e$ is an edge that includes $v$, and $0$ otherwise.
The valuation of a receiver $r_{l}$, $l \in L$, is $v_{r_{v}}(\theta_l)=1$ and $v_{r_{l}}(\theta_{l'})=0$ for each $l' \neq l$.
Moreover the valuation of a receiver $\{r_i\}$, $i \in \{1,\dots,m+1\}$ is $v_{r_{i}}(\theta)=(1-\eta)/z$ for each state $\theta \in \Theta$.
Finally, the prior is uniform over all the states.

Let $L^*$ be the minimum vertex cover and $z^*$ be its size.
We show how to build a vertex cover of size at most $z^*(1+\epsilon)$ from the solutions to the signaling problems instantiated with $z \in \{1,\dots,\rho\}$.
For all $z\in \{1,\dots,\rho\}$, given a signaling scheme we recover a vertex cover $L(z)$ as follows.
Take the posterior with larger sender's utility and add to the vertex cover $L(z)$  all the vertexes $l \in L$ such that the receiver $r_l$ has valuation at least $1/z (1-\frac{3}{4}\eta)$. Then, for each edge that is not covered, we add to $L(z)$ one arbitrary adjacent vertex.
It is easy to see that the resulting solution $L(z)$ is a vertex cover.
Finally, the algorithm returns the smallest among the vertex covers $L(z)$, $z \in \{1,\dots,\rho\}$. 

We show that the vertex cover $L(z^*)$ has size at most $z^*(1+\epsilon)$, concluding the proof. 
First, we show that the optimal solution of the signaling problem is at least $\frac{m_{z^*}}{z^*}(1-\frac{1}{2}\eta)$. 
Consider the signaling scheme with two signals $s_1$ and $s_2$ with $\phi_{\theta_l}(s_1)=1$ for each $l \in L^*$ and $\phi_{\theta_l}(s_2)=1$ for each $l \notin L^*$.
In the posterior induced by $s_1$ the revenue is at least $\frac{m_{z^*}}{z^*}$ since all the receivers $r_{e,i}$,  have expected at least $1/z^*$, while all the receivers $\{r_l\}_{l \in L^*}$ have utility at least $1/z^*$. Hence, there are at least $z^*+ \rho|E|=m_{z^*}+1$ agents with valuation at least $1/z^*$.
Moreover, in the posterior induced by $s_2$ the revenue is at least $(1-\eta)\frac{m_{z^*}}{z^*}$ since all the receivers $\{r_i\}_{i \in\{1,\dots,m_{z^*+1}\}}$ have expected valuation $(1-\eta)/z^*$.
Since $z^*\ge |E|/3$ and $\rho=\frac{2}{3}|E|$ signal $s_1$ is sent with probability at least $z^*/\rho \ge \frac{1}{2}$ and the solution has value at least $\frac{1}{2}\frac{m_{z^*}}{z^*}+ \frac{1}{2}(1-\eta)\frac{m_{z^*}}{z^*}=\frac{m_{z^*}}{z^*}(1-\frac{1}{2}\eta)$.
Hence, a $1-\delta$ approximation algorithm for the signaling problem must return a signaling scheme with value at least $\frac{m_{z^*}}{z^*}(1-\frac{1}{2}\eta)(1-\delta)\ge \frac{m_{z^*}}{z^*}(1-\frac{3}{4}\eta)$.
Since the expected revenue is of the signaling scheme is at least $\frac{m_{z^*}}{z^*}(1-\frac{3}{4}\eta)$, this signaling scheme sends a signal that induces a posterior $\xi \in \Delta_\Theta$ with revenue at least $\frac{m_{z^*}}{z^*}(1-\frac{3}{4}\eta)$.
This implies that there are at least $z^*$ receivers $r_{l}$ with utility greater or equal to $\frac{1}{z^*}(1-\frac{3}{4}\eta)$ and that the utility of all the receivers $r_{e,i}$ is at least $\frac{1}{z^*}(1-\frac{3}{4}\eta)$.
We show that our algorithm recovers a vertex cover with size at most $z^*(1+7 \eta)= z^*(1+\epsilon)$ from this posterior.
Consider the set of vertexes $L^1$ with utility at least $\frac{1}{z^*}(1-\frac{3}{4}\eta)$. This set has size at most $z^*/(1-\frac{3}{4}\eta)$ since in each state only one receiver $r_l$ has valuation $1$ and all the other receivers $r_l'$, $l'\neq l$, have valuation $0$.
Consider the set $ L^2 =L \setminus L^1$ of vertexes not in this set.
We have that $\sum_{l \in L^2} \xi (\theta_l) \le \frac{3}{4}\eta$ since $\sum_{l \in L^1} \xi (\theta_l)= \sum_{l \in L^1} \xi^\top v_{r_l}  \ge (1-\frac{3}{4}\eta)$.
Let $\bar E$ be the set of edges not covered by $L^1$. Since each vertex has three edges, we have that $\sum_{e \in \bar E} \xi^\top v_{r_{e,i}}\le 3\frac{3}{4}\eta$ for each $i \in \{1,\dots,m_{z^*}+1\}$. Moreover, since for each edge in $\bar E$, $\xi^\top v_{r_{e,i}}\ge (1-\frac{3}{4}\eta)/z^*$, we have that $|\bar E| \le \frac{3\frac{3}{4}\eta} {(1-\frac{3}{4}\eta)/z^*}$.
Then, the vertex cover build by the algorithm for $z=z^*$ includes at most \[ z^*( \frac{1}{1-\frac{3}{4}\eta} + \frac{3\frac{3}{4}\eta}{1-\frac{3}{4}\eta})\le z^*(1+\frac{9}{4} \eta)  (1+\frac{3}{2} \eta)\le z^* (1+7 \eta)=z^*(1+\epsilon).\]

\end{proof}

\lemsupport*
\begin{proof}
To prove the result we show that, given a signaling scheme $\gamma$ that induces two posteriors $\xi,\xi' \in \Xi_{\pi}$ for a $\pi \in \Pi_{m+1}$, we can recover a signaling scheme $\gamma^*$ with at least the same revenue that replaces the two posteriors $\xi$ and $\xi'$ with a convex combination of them.
 Let $\xi, \xi' \in \Xi_{\pi}$ be two elements belonging to the support of an optimal signaling scheme $\gamma$. In order to show the result we introduce a posterior probability $\xi^*$ as follow: $\xi^* = z \xi_1 + (1- z) \xi_2$ with $ z=\gamma(\xi_1)/(\gamma(\xi_1)+\gamma(\xi_2) )$. Since $\Xi_{\pi}$ is a convex polytope each convex combination of a subset of its elements belongs to it. Hence, $\xi^* \in \Xi_{\pi}$. Moreover, we will define a new signaling scheme $\gamma^*$ as follow: $\gamma^*(\xi^*) = \gamma(\xi_1) + \gamma(\xi_2)$ and $\gamma^*(\xi_1) = \gamma^*(\xi_2) = 0$ while $\gamma^*(\xi_i) = \gamma(\xi_i)$ $\forall i \ne 1,2$. To conclude the proof, we observe that the two signaling schemes gain the same revenue. Indeed, by linearity we have: $\gamma(\xi^*)  \textsc{Rev}(V, \xi^*) = ( \gamma(\xi_1)+\gamma(\xi_2))  \textsc{Rev}(V, z \xi_1 + (1-z) \xi_2) = \gamma(\xi_1) \textsc{Rev}(V, \xi_1)+\gamma(\xi_2) \textsc{Rev}(V, \xi_2) $ .
 %\textcolor{blue}{but $|\Xi_\pi \cap \textnormal{supp}(\gamma^*)| \leq 1$}.
\end{proof}

\lemmaKnownD*
\begin{proof}
Given a tuple $\pi \in \Pi_{m+1}$ and $\xi \in \Xi_{\pi}$, we define $x_{\pi}(\theta) := \gamma(\xi)\hspace{0.2mm} \xi(\theta)$ for each $ \theta \in \Theta$ as the posterior probability multiplied by the probability of state $\theta$ in $\xi$.
Notice that by Lemma \ref{lem:support} there is at most one $\xi \in \Xi_{\pi}$ belonging to the support of an optimal $\gamma$ for each possible tuple $\pi \in \Pi$. Thus, we can represent our optimization problem with the following LP.

\begin{subequations}\label{eqn:LP_kfixed}
	\begin{align}
	\max_{x \in [0,1]^{|\Pi_{m+1}||\Theta|}} & \sum_{\pi=(i_1,\dots,i_{m+1}) \in \Pi_{m+1} } \sum_{\theta \in \Theta}  x_{\pi}(\theta) \sum_{j=1}^{m} j \,  v_{i_{j + 1}}(\theta) (\lambda_{j}-\lambda_{j+1}) 	\\
	s.t. \ 	& \sum_{\pi \in \Pi_{m+1}} x_{\pi}(\theta) = \mu_{\theta} \hspace{7cm} \forall \theta \in \Theta\\
	& \sum_{\theta \in \Theta} x_{\pi}(\theta) [v_{i_j}(\theta)-v_{i_{j+1}}(\theta) ]\ge 0 \hspace{5mm}  \forall \pi=(i_1,\dots,i_{m+1}) \in \Pi_{m+1}, j \in \{1,\dots,m\} .
	\end{align}
\end{subequations}

 Note that LP \ref{eqn:LP_kfixed} is solvable in polynomial time as long as $m$ is fixed.
 To conclude the proof, we show that from a solution of LP \ref{eqn:LP_kfixed} we can always recover a signaling scheme setting $\gamma(\xi_\pi)=\sum_{\theta \in \Theta} x_{\pi}(\theta)$ for each $\pi \in \Pi_{m+1}$ and $\xi_\pi(\theta)=x_\pi(\theta)/ \gamma(\xi_\pi) $ for each $\pi \in \Pi_{m+1}$ and $\theta \in \Theta$ if $\gamma(\xi_\pi) \ne 0$. Moreover, given a signaling scheme $\gamma$ we can compute a solution to LP \ref{eqn:LP_kfixed} using the same relation. Finally, we show that an optimal solution to LP \ref{eqn:LP_kfixed} provides the same value of the relative distribution $\gamma$. In particular, we have 
 \begin{subequations}
 	\begin{align*}
 	\sum_{\pi \in \Pi_{m+1}} \hspace{-0.15cm}\gamma(\xi_\pi) \sum_{\theta \in \Theta} \sum_{j=1}^{m} j \, \xi_\pi(\theta) v_{i_{j + 1}}(\theta) (\lambda_{j}-\lambda_{j+1})  &=\hspace{-0.15cm} \sum_{\pi=(i_1,\dots,i_{m+1}) \in \Pi_{m+1}} \sum_{\theta \in \Theta} \gamma(\xi_\pi) \xi_\pi(\theta) \sum_{j=1}^{m} j \, v_{i_{j + 1}}(\theta) (\lambda_{j}-\lambda_{j+1}) \\
 	&= \hspace{-0.15cm}\sum_{\pi=(i_1,\dots,i_{m+1}) \in \Pi_{m+1}} \sum_{\theta \in \Theta} x_\pi(\theta) \sum_{j=1}^{m} j \, v_{i_{j + 1}}(\theta) (\lambda_{j}-\lambda_{j+1}) 
 	\end{align*}
 	\end{subequations}
\end{proof}

\lemmaKnownfixd*
\begin{proof}
First, we observe that for each $\xi \in \Delta_{\Theta}$ there exists a tuple $\pi \in \Pi_n$ such that $\xi \in \Xi_{\pi}$, this easily follow from the fact that $\bigcup_{\pi \in \Pi_n} \Xi_{\pi} = \Delta_{\Theta}$.
As observed before in such regions the revenue is a linear function.
Thus, it is possible to decompose each posterior $\xi \in \Xi_{\pi}$ by Caratheodory's theorem as a convex combination of the vertexes of $\Xi_{\pi}$ without decreasing the revenue.
Formally, for each $\pi \in \Pi_n$ and each posterior $\xi \in \Xi_\pi$, there exists a distribution $\gamma_{\xi} \in \Delta_{V(\Xi_{\pi})}$ such that
\begin{equation*}
\xi(\theta)= \sum_{\tilde{\xi} \in \Xi^*} \gamma_{\xi}(\tilde{\xi}) \tilde{\xi}(\theta) \hspace{2mm} \forall \theta \in \Theta.
\end{equation*}
We show that such a decomposition does not affect the final revenue. Indeed, by linearity we get the following:
\begin{equation*}
\sum_{\tilde{\xi} \in V(\Xi_\pi) } \gamma_\xi(\tilde{\xi}) \textsc{Rev}(V, \tilde{\xi}) = \textsc{Rev} \Big (V,  \sum_{\tilde{\xi} \in V( \Xi_\pi ) } \gamma_\xi( \tilde{ \xi})  \tilde{\xi} \Big)
= \textsc{Rev}(V, \xi).
\end{equation*}
To conclude the proof, we show that given the optimal distribution $\gamma$, we can recover a distribution $\gamma^* \in \Delta_{\Xi^*}$ with the same revenue. In particular, $\gamma^* \in \Delta_{\Xi^*}$ is such that:
\begin{equation*}
\gamma^*(\tilde{\xi})=\sum_{\xi \in supp(\gamma)} \gamma(\xi) \gamma_\xi(\tilde{\xi}) \hspace{4mm} \forall \tilde{\xi} \in \Xi^*.
\end{equation*}
Since $\gamma$ satisfies the consistency constraints, it easy to see that also $\gamma^* \in \Delta_{\Xi^*}$ satisfies the consistency constraints.
Moreover, the two distribution provide the same revenue. Indeed, we have
\begin{align*}
\sum_{\tilde{\xi} \in \Xi^*}\gamma^*(\tilde{\xi}) \ \textsc{Rev}(V, \tilde{\xi})
& = \sum_{\xi \in supp(\gamma)} \gamma(\xi) \sum_{\tilde{\xi} \in \Xi^* }
\gamma_\xi(\tilde{\xi}) \textsc{Rev} ( V, \tilde{\xi} ) \\
& = \sum_{\xi \in supp(\gamma)} \gamma(\xi) \ \textsc{Rev}(V, \xi),
\end{align*}
and $OPT=OPT_{\Xi^*}$

\end{proof}

\thmknownfixd*
\begin{proof}
We first observe that the vertexes of each region $\Xi_\pi$ are identify by the intersection of $d-1$ linear independent hyperplanes for each $\pi \in \Pi_n$. Moreover, we note that each of these vertexes is identified by a subset of the $O(n^2)$ constraints $\xi^\top v_i \ge \xi^\top v_j$ for each $i\neq j \in \N$ and the $d$ constraint $\xi(\theta)\ge 0$ for each $\theta \in \Theta$. Hence, the total number of vertexes defining the previous discussed regions will be equal to $|\Xi^*| = O((n^2+d)^{d-1})$. Finally, we notice that, as long as d is a fixed parameter, it is possible to find an optimal signaling scheme in polynomial time solving LP \ref{eq:lp_signaling} with set of posteriors $\Xi^*$.
\end{proof}

\uniformVertices*
\begin{proof}
	
	First, we show that if a posterior $\xi \in \Xi_\pi$ satisfies (i) and (ii), then $\xi \in V(\Xi_\pi)$, \emph{i.e.}, it is a vertex of $\Xi_\pi$.
	In particular, $\xi$ satisfies the linear independent equality $ \delta_{\theta_{k_j}} \xi(\theta_{k_{j}}) = \delta_{\theta_{k_{j+1}}} \xi(\theta_{k_{j+1}}) $ for each $j \in [\ell-1]$ \footnote{Given $n \in \mathbb{N}_{>0}$ we denote with $[n]=\{ 1,.., n \}$}
	Moreover, it satisfies $\delta_{\theta_{k_j}} \xi(\theta_{k_{j}})=0$ for each $j \in \{\ell+1,\dots,d\}$ and the simplex equality $\sum_{\theta \in \Theta}\xi(\theta)=1$.
	Hence, $\xi \in \Xi_\pi$ is at the intersection of $d$ linear independent hyperplanes defining $\Xi_\pi$ and it is a vertex of $\Xi_\pi$.
	To conclude the proof, we show that each vertex $\xi \in \Xi_\pi$ satisfies (i) and (ii).
	In particular, we show that given a posterior $\xi \in \Xi_\pi$ such that $\delta_{\theta_{k_{j^*}}} \xi(\theta_{k_{j^*}}) > \delta_{\theta_{k_{{j^*}+1}}} \xi(\theta_{k_{{j^*}+1}})> 0 $ for a $j^* \in [d]$, \emph{i.e.}, it does not satisfies (i) and (ii), the posterior is at the the intersection of at most $d-1$ linear independent hyperplanes.
	Consider the hyperplanes $\delta_{\theta_{k_{j}}} \xi(\theta_{k_{j}}) = \delta_{\theta_{k_{{j}+1}}} \xi(\theta_{k_{{j}+1}}) $ for $j\le j^*-1$.
	Notice that $\xi$ satisfies at most all the $j^*-1$ inequalities of this kind. \footnote{Notice that the equality $\delta_{\theta_{k_{{j}}}} \xi(\theta_{k_{{j}}})=\delta_{\theta_{k_{{j'}}}} \xi(\theta_{k_{{j'}}}) $ with $j>j'$ and$|j -j'|\ge 2$ is linear dependent from the equalities $\delta_{\theta_{k_{{\bar j}}}} \xi(\theta_{k_{{\bar j}}})=\delta_{\theta_{k_{{\bar j+1}}}} \xi(\theta_{k_{{\bar j+1}}}) $ for each $\bar j \in \{ j,\dots,j'-1\}$.}
	Moreover, consider all the $j>j^*$ with $\delta_{\theta_{k_{j^*}}} \xi(\theta_{k_{j^*}})>0$. Let $j^{**}$ be the largest $j$ that satisfies this condition. By a similar argument as above, we can show that there are at most $j^{**}-j^*-1$ linear inequalities $\delta_{\theta_{k_{{j}}}} \xi(\theta_{k_{{j}}})=\delta_{\theta_{k_{{j+1}}}} \xi(\theta_{k_{{j+1}}}) $ with $j^*+1\le j\le j^{**}-1$.
	Finally, for all the $j>j^{**}$, the equality $\delta_{\theta_{k_{{j}}}} \xi(\theta_{k_{{j}}}) =0$ is satisfied.
	Hence, including the simplex constraint there are at most $j^*-1+j^{**}-j^*-1+d-j^{**}+1=d-1$ linear independent equalities, concluding the proof.
	
\end{proof}

\dinprog*

\begin{proof}
	Let $f(v,j)$ be the revenue when $j \in [n]$ bidders have expected valuation $v \in [0,1]$ and all the other bidders have expected valuation $0$.
	Moreover, given a set $E \subseteq \mathbb{R}$ and an $x \in \mathbb{R}$, let $\lfloor x \rfloor_{E}$ be equal to the largest element $e \in E$ such that $e\le x$. Similarly, we define $\lceil x \rceil_{E}$ be equal to the smallest element $e \in E$ such that $e\ge x$.

	\begin{algorithm} 
		\caption{Dynamic programming algorithm in the proof of Lemma~\ref{lem:dp} }\label{alg:cap}
		\begin{algorithmic}[1]
			\REQUIRE $ \epsilon >0 $
			\STATE  $c\gets \lceil1/ \epsilon \rceil$
			\STATE  $E \gets  \{ i /c \}_{i=0}^{c} $
			\STATE $G \gets \cup_{\theta \in \Theta} \{  \delta_\theta i/c \}_{i=0}^{c} $
			\STATE initialize empty matrices $M$ and $\Theta$ with dimension $cd \times d \times c \times n$
			%\STATE $N \gets n$
			\FOR{$v \in E$}
					\FOR{$w \in E, w \ge \lfloor v/\delta_{\theta_i}\rfloor_E $}
						\STATE $M(v,1,w,|\N_{\theta_1}|)\gets - y_{\theta_1} v/\delta_{\theta_1} $
						\STATE $\Theta(v,1,w,|\N_{\theta_1}|) \gets \{\theta_1\}$
					\ENDFOR
			\FOR{$i \in [d]$,$w \in E$,$j \in [n]$}
			\IF{$M(v,i-1,w,j)\ge  M(v,\theta_{i-1},w-\lceil v/\delta_{\theta_i}\rceil_E,j-|\N_{\theta_i}|))-y_{\theta_i} v/\delta_{\theta_i}$}
			\STATE $M(v,i,w,j) \gets M(v,i-1,w,j)$
			\STATE $\Theta(v,i,w,j) \gets \Theta(v,i-1,w,j)$
			\ELSE
			\STATE $M(v,i,w,j) \gets M(v,i-1,w-\lceil v/\delta_{\theta_i}\rceil_E,j) - y_{\theta_i} v/\delta_{\theta_i}$
			\STATE $\Theta(v,i,w,j) \gets \Theta(v,i-1,w-\lceil v/\delta_{\theta_i}\rceil_E,j) \cup \{\theta_i\}$
			\ENDIF 
			\ENDFOR
			\ENDFOR
			\STATE $(\hat v,\hat j, \hat w) \gets \argmax_{v \in E,j \in [n],w \in E} f(v,j)+M(v,d,w,j)$
			\STATE $\hat \Theta \gets \Theta(\hat v,d,\hat w, \hat j) $ 
			\STATE $w_{real}=\sum_{\theta \in \hat \Theta} \hat v/\delta_\theta$,
			\FOR{$\theta \in \hat \Theta$}
			\STATE $\hat \xi(\theta)=\frac{\hat v}{\delta_\theta w_{real}}$
			\ENDFOR
			\RETURN $\hat \xi$
		\end{algorithmic}
	\end{algorithm}
	
	We show that Algorithm \ref{alg:cap} provides the desired guarantees.
	It is easy to see that the algorithm runs in polynomial time.
	Let $\xi^*$ be the optimal solution to $\max_{\xi \in \Xi^*} \textsc{Rev}(V, \xi) - \sum_{\theta \in \Theta}    y_{\theta} \, \xi(\theta) $.
	Notice that by the definition of $\Xi^*$ there exists a subset of states $\Theta^* \subseteq \Theta$ and a value $v^*$ such that $\xi^*(\theta) \delta_{\theta}=v^*$ for each $\theta \in \Theta^*$ and $\xi^*(\theta)=0 $ for each $\theta \notin \Theta^*$.
	Our first step is to show that this solution $\xi^*$ corresponds to a feasible solution to the algorithm.
	Let take $v=\lfloor v^*\rfloor_G$ and $w=\sum_{\theta \in \Theta^*} \lfloor v/\delta_{\theta}\rfloor_E$.
	Formally, we show that $\Theta^*$ is a feasible subset of states for $\Theta(v,d,w,\sum_{\theta \in \Theta^*} |\N_\theta|)$.
	In particular, it is sufficient to prove that $w = \sum_{\theta \in \Theta^*} \lfloor v/\delta_\theta \rfloor_E \le \sum_{\theta \in \Theta^*} v/\delta_\theta\le \sum_{\theta' \in \Theta^*} v^*/\delta_\theta =1$.
	
	Since this solution is feasible, it provides a lower bound on the value $\max_{v \in G,j \in [n],w \in E}f(v,j)+M(v,d,w,j)$.
	Let $\theta^* \in \Theta^*$ be the state in $\Theta^*$ with smallest $\delta_\theta$.
	First, we provide a bound on $v^*$.
	It holds $v^*\le \delta_{\theta^*}$, otherwise $\xi^*(\theta)=v^*/ \delta_{\theta^*}>1$.
	Moreover, $v^* \ge \delta_{\theta^*}/d$, otherwise $\sum_{\theta \in \Theta^*} v^*/\delta_\theta<1$.
	This implies that for each $\theta \in \Theta^*$,  $v/\delta_\theta=\lfloor v^*\rfloor_G / \delta_{\theta} \ge (v^*-d\epsilon \delta_{\theta^*}) /\delta_\theta \ge v^* /\delta_\theta -d\epsilon $ and $v \ge v^*-\epsilon d$.
	%Thus, $\sum_{\theta \in \Theta^*} \lfloor v/\delta_\theta \rfloor_E \ge \sum_{\theta \in \Theta^*}  v/\delta_\theta -\epsilon d  \ge \sum_{\theta \in \Theta^*}  v^*/\delta_\theta -\epsilon d-d^2\epsilon\ge  \sum_{\theta \in \Theta^*}  v^*/\delta_\theta -2d^2\epsilon $.
	Now, we can provide our lower bound on $\max_{v \in G,j \in [n],w \in E}f(v,j)+M(v,d,w,j)$.
	In particular, the solution that takes states $\Theta^*$, $v=\lfloor v^*\rfloor_G$, and $w=\sum_{\theta \in \Theta^*} \lfloor v/\delta_{\theta}\rfloor_E$ has value at least and has value at least 
	\begin{subequations}
		\begin{align*}
			f(v,\sum_{\theta \in \Theta^*} |\N_\theta|) -\sum_{\theta \in \Theta^*} \frac{y_\theta v}{\delta_\theta  }&\ge f(v^*,\sum_{\theta \in \Theta^*} |\N_\theta|) - \epsilon dm-\sum_{\theta \in \Theta^*} \frac{y_\theta v}{\delta_\theta }\\
			&\ge f(v^*,\sum_{\theta \in \Theta^*} |\N_\theta|) -\sum_{\theta \in \Theta^*} y_\theta  \xi(\theta) - \epsilon d m - d^2 \beta \epsilon\\
			& = 	\max_{\xi \in \Xi^*} \,\, \textsc{Rev}(V, \xi) - \sum_{\theta \in \Theta}    y_{\theta} \, \xi(\theta)- \epsilon d m - d^2 \beta \epsilon ,
		\end{align*}
	\end{subequations}
where the first inequality comes from Lipschitz continuity of $f(\cdot,x)$, \emph{i.e.}, $f(v^*,x)-f(v,x)\le m |v^*-v|$ for each $x \in \mathbb{N}$ and $v=\lfloor v^*\rfloor_G \ge v^* -\epsilon d$, while the second inequality comes from $-\frac{y_\theta v}{\delta_\theta  } \ge -y_\theta (v^*/\delta_\theta -d \epsilon )\ge -y_\theta v^*/\delta_\theta -d \beta \epsilon $ for each $\theta \in \Theta^*$.
	
	To conclude the proof, we show that from a solution $(\hat v,\hat j,\hat w)$ and $\hat \Theta$ the algorithm find a posterior $\hat \xi$ with value at least $f(\hat v,\hat j)+M(\hat v,d,\hat w,\hat j)- \epsilon d m-2\beta d \epsilon$.
	First, we bound the value of $w_{real}$. In particular, $ w_{real} \le 1+ \epsilon d$ since for each state $\theta$, $\hat v/\delta_\theta-\lfloor \hat v/\delta_\theta\rfloor_E \le \epsilon$ and $\hat w\le 1$.
	%	 and in $\xi$ the valuations of the buyers are at least $\hat v/(1+ \epsilon d)$.
	Hence, in the posterior $\hat \xi$ all the bidders have valuation at least $\hat v/(1+\epsilon d) \ge \hat v- \epsilon d$ and $\textsc{Rev}(V, \hat \xi)=f(\hat v-\epsilon d, \sum_{\theta \in \hat \Theta}|\N(\theta))\ge f(\hat v, \sum_{\theta \in \hat \Theta}|\N(\theta)|) -\epsilon d m$ by the Lipschitz continuity of $f(\cdot,x)$.
	Now, we consider the component $M(\hat v,d,\hat w,\hat j)$. In particular, we show that $\sum_{\theta \in \hat \Theta} -\hat \xi(\theta) y_\theta \ge M(\hat v,d,\hat w,\hat j)-2\beta d \epsilon$.
	Since $y_\theta\le 0$  for each $\theta$, it holds
	\[\sum_{\theta \in \hat \Theta} -\hat \xi(\theta) y_\theta = \sum_{\theta \in \hat \Theta} - \frac{\hat v}{\delta_\theta w_{real}} y_\theta = M(\hat v,d,\hat w,\hat j)/w_{real} \ge  M(\hat v,d,\hat w,\hat j)/(1+d\epsilon) \ge M(\hat v,d,\hat w,\hat j) + 2\beta d \epsilon,  \]
	where the last inequality comes from $M(\hat v,d,\hat w,\hat j) \le \beta w_{real}\le \beta (1+d\epsilon)$ and $ 1/(1+d\epsilon) \ge 1-d\epsilon $ .
	
	To conclude, the value of the solution $\hat \xi$ is an additive $(\epsilon d m + d^2 \beta \epsilon+\epsilon d m+2\beta d \epsilon)$-approximation.
	For $\epsilon$ small enough, we obtain the desired approximation.
\end{proof}

\fptasSingleMinded*
\begin{proof}
Our FPTAS is described in Algorithm~\ref{alg:bisection}.

\begin{algorithm}[H]\caption{FPTAS in the proof of Theorem \ref{thm:fptasSingleMinded}}
	\textbf{Input:} parameter of the relaxed LP $\beta$, approximation factor of the approximation oracle $\lambda$, error $\eta$.
	\begin{algorithmic}[1]
		\STATE \textbf{Initialization}: $\rho_1\gets0$, $\rho_2\gets1$,  $H \gets\emptyset$,$H^* \gets \emptyset$.
		\WHILE{ $\rho_2-\rho_1>\eta$}
		\STATE $\rho_3\gets (\rho_1+\rho_2)/2$
		\STATE $H \gets\{\text{\normalfont posteriors relative to the violated constraints returned by the ellipsoid method on } \circled{F} \newline
		 \text{\normalfont with objective } \rho_3 \,\text{\normalfont and approximation error}\, \delta \}$
		\IF  {unfeasible}
		\STATE $\rho_1\gets \rho_3$
		\STATE $H^*\gets H$
		\ELSE{
			\STATE $\rho_2 \gets \rho_3$}
		\ENDIF
		\STATE $(\gamma,z) \gets$ solution to LP \ref{eq:lp_signaling_reduced} with only posteriors in $H^*$
		\STATE \textbf{return} the solution $\bar \gamma$ corresponding to the solution of the relaxed problem $(\gamma,z)$
		\ENDWHILE
	\end{algorithmic}	\label{alg:bisection}
\end{algorithm}

We start providing the following relaxation to LP \ref{eq:lp_signaling} for a value $\beta \in \mathbb{R}_+$ defined in the following.

\begin{subequations}\label{eq:lp_signaling_relax}
	\begin{align}
	\max_{\gamma \in \Delta_{\Xi^*},z\le 0} & \,\, \sum_{\xi \in \Xi} \gamma(\xi) \, \textsc{Rev}(\V,\xi) +   \beta z \quad \text{s.t.} \\
	& \sum_{\xi \in \Xi^*} \gamma(\xi) \, \xi(\theta)- z  \ge \mu_\theta & \forall \theta \in \Theta. 
	\end{align}
\end{subequations}

Given a solution to $\gamma$, $z$ to LP \ref{eq:lp_signaling_relax}, we can find an approximate solution to LP \ref{eq:lp_signaling} as follows.
First, notice that by the optimality of $\gamma$, $z$, we have $\sum_{\xi \in \Xi} \gamma(\xi) \, \textsc{Rev}(\V,\xi) +   \beta z  \ge 0$ and $z\ge - m/\beta$.
Let $\bar \mu=\sum_{\xi \in \Xi^*} \gamma(\xi) \, \xi$ be the mean of $\gamma$.
We have that $ |\bar \mu_\theta- \mu_\theta| \le dm/\beta$ for each $\theta \in \Theta$.
Consider a distribution $\bar \gamma$ such that 
\begin{align}\label{eq:modification}
\bar \gamma(\xi)=\gamma(\xi)(1-dm/\beta)+ \sum_\theta \mathbb{I}_{\xi(\theta)=1} [\mu_{\theta}-\sum_\xi\gamma(\xi) \xi(\theta) (1-dm/\beta)],
\end{align} 
where $\mathbb{I}_{\xi(\theta)=1}=1$ iff $\xi(\theta)=1$  and $0$ otherwise.
$\bar \gamma$ is a feasible solution to LP \ref{eq:lp_signaling} since $\sum_{\xi \in \Xi^*} \bar \gamma(\xi) \mu (\xi) = \sum_{\xi \in \Xi^*} \bar \gamma(\xi) \xi (\theta) (1-dm/\beta)+ 
 [\mu_{\theta}-\sum_{\xi \in \Xi^*} \gamma(\xi) \xi(\theta) (1-dm/\beta)]=\mu_\theta$ and $\bar \gamma(\xi)\ge 0$ for each $\xi \in \Xi^*$. Moreover, it has value at least $OPT_{\Xi^*}-dm^2/\beta$ since the distribution $\gamma$ is scaled by a factor $(1-dm/\beta)$ and $OPT_{\Xi^*}\le m$.

Hence, to provide an approximation to LP \ref{eq:lp_signaling}, it is sufficient to provide an approximation to LP \ref{eq:lp_signaling_relax}  for a sufficiently large $\beta$.
Since LP~\ref{eq:lp_signaling_relax} has an exponential number of variables, the algorithm works by applying the ellipsoid method to the following dual problem.

\begin{subequations}\label{eq:dual_lp_signaling_relax}
	\begin{align}
	\min_{y\le 0, t} & \, \sum_{\theta \in \Theta} y_\theta \mu_\theta + t  \quad \text{s.t.} \\
	& \sum_{\theta \in \Theta} y_\theta \, \xi(\theta) + t \geq \textsc{Rev}(V,\xi) & \forall \xi \in \Xi^*\\
	& \sum_\theta y_\theta \ge -\beta, 
	\end{align}
\end{subequations}
where the dual variables are $y \in \mathbb{R}^d_{-} $ and $t \in \mathbb{R}$.
Instead of an exact separation oracle, we use an approximate separation oracle that employs Algorithm~\ref{alg:cap} with a suitably-defined approximation $\lambda>0$.
We use a binary search scheme to find a value $\rho^\star\in [0,1]$ such that the dual problem with objective $\rho^\star$ is unfeasible, while the dual with objective $\rho^\star+\eta$ is \emph{approximately} feasible, for some $\eta \geq 0$ defined in the following.
The algorithm requires $\log(\eta)$ steps and, at each step, it works by determining, for a given value $\rho_3$, whether there exists a feasible solution for the following feasibility problem that we call \circled{F}:
\begin{subequations}
	\begin{align} \label{lp:privatedual1}
	& \sum_{\theta \in \Theta} y_\theta \mu_\theta + t \le \rho_3\\
	&  \sum_{\theta \in \Theta} y_\theta \, \xi(\theta) + t \geq \textsc{Rev}(V,\xi) & \forall \xi \in \Xi^*\\
	&\sum_\theta y_\theta \ge -\beta\\
	&y_\theta \le 0 &\forall \theta \in \Theta.
	\end{align}
\end{subequations}

At each iteration of the bisection algorithm, the feasibility problem \circled{F} is solved via the ellipsoid method.
To do so, we need a separation oracle.
We focus on an approximate separation oracle that returns a violated constraint. 
The oracle is implemented as follows. Given a point $(y,t)$, first it check if all the $y_\theta$ are greater than $0$, $\sum_{\theta} y_\theta\ge -\beta$, and $\sum_{\theta \in \Theta} y_\theta \mu_\theta + t \le \rho_3$. If it is not the case, it returns a violated constraint.
Otherwise, it uses Algorithm \ref{alg:cap} to approximate 
\begin{equation}\label{eq:separation_relax}
\max_{\xi \in \Xi^*} \textsc{Rev}(V, \xi) - \sum_{\theta \in \Theta}    y_{\theta} \, \xi(\theta)   
\end{equation}
Notice, that we guarantee that Algorithm \ref{alg:cap} is called with $y_\theta \ge -\beta$ for each $\theta \in \Theta$.
Let $\xi$ be the returned posterior.
If there returned posterior has value at least $t$, the separation oracle returns the constraint relative to posterior $\xi$. Otherwise, it returns feasible.
The bisection procedure terminates when it determines a value $\rho^\star$ such that on \circled{F} the ellipsoid method returns unfeasible for $\rho^\star$, while returning feasible for $\rho^\star+\eta$.
Then, the algorithm solves a modified primal LP \ref{eq:lp_signaling_reduced} with only the subset of posteriors in $H^*$, where $H^*$ is the set of posteriors relative to the violated constraints returned by the ellipsoid method applied on the unfeasible problem with objective $\rho^*$.
Finally, it computes a solution $\bar \gamma$ from the solution $\gamma$ of LP \eqref{eq:lp_signaling_reduced} using \ref{eq:modification}.

Now, we prove the approximation guarantees of the algorithm.
The algorithm finds a $\rho^*$ such that the problem is unfeasible, \emph{i.e.}, the value of $\rho_1$ when the algorithm terminates, and a value smaller than or equal to $\rho^*+\eta$ such that the ellipsoid method returns feasible, \emph{i.e.}, the value of $\rho_2$ when the algorithm terminates.
In particular, we show that $OPT \le \rho^*+\beta+\delta$, where $OPT$ is the value of LP \ref{eq:lp_signaling_relax}.
Since, the bisection algorithm returns that \circled{F} is feasible with objective $\rho^*+\eta$, it finds a solution $(y,t)$ such that the approximate separation oracle did not find a violated constraint.
We show that $(y,t)$ is a solution to the following LP.

\begin{subequations}\label{lp:privatedual2}
	\begin{align} 
	& \sum_{\theta \in \Theta} y_\theta \mu_\theta + t \le \rho^*+\eta\\
	&  \sum_{\theta \in \Theta} y_\theta \, \xi(\theta) + t \geq \textsc{Rev}(V,\xi) -\lambda& \forall \xi \in \Xi^*\label{eq:modifieddual1}\\
	&\sum_\theta y_\theta \ge -\beta\\
	&y_\theta \le 0 &\forall \theta \in \Theta.
	\end{align}
\end{subequations}

This holds because we have shown that, when the separation oracle returns feasible, it holds $\max_{\xi \in \Xi^*} [\textsc{Rev}(V, \xi) - \sum_{\theta \in \Theta}    y_{\theta} \, \xi(\theta)]\le t+\lambda$ by the approximation guarantees of Algorithm \ref{alg:cap}, implying that all the Constraints \eqref{eq:modifieddual1} are satisfied.
Moreover, when the separation oracle returns feasible all the other constraints are satisfied.
Then, by strong duality the value of the following LP is at most $\rho^*+\eta$.

\begin{subequations}\label{eq:lp_mod}
	\begin{align}
	\max_{\gamma \in \Delta_{\Xi^*},z\le 0} & \,\, \sum_{\xi \in \Xi^*} \gamma(\xi) \, (\textsc{Rev}(V,\xi)-\lambda) + \beta z \quad \text{s.t.} \\
	& \sum_{\xi \in \Xi} \gamma(\xi) \, \xi(\theta)-z\ge \mu_\theta & \forall \theta \in \Theta.
	\end{align}
\end{subequations}

Notice that any solution to LP \ref{eq:lp_signaling} is also a feasible solution to the previous modified problem.
Since in any feasible solution $\sum_{\xi \in \Xi^*} \gamma(\xi)  =1 $ and LP \ref{eq:lp_mod} has value at most $\rho^* +\eta$, then $OPT \le \rho^* +\eta +\lambda$.

Let $H^*$ be the set of posteriors relative to the constraints returned by the ellipsoid method run with objective $\rho^\star$.
Since the ellipsoid method with the approximate separation oracle returns unfeasible, by strong duality LP \ref{eq:lp_signaling_relax} with only the variables $\gamma(\xi)$ relative to constraints in $H^*$ has value at least $\rho^*$. Moreover, since the ellipsoid method guarantees that $H^*$ has polynomial size, the LP can be solved in polynomial time.
Hence, solving the following LP, \emph{i.e.}, the primal LP \ref{eq:dual_lp_signaling_relax} with only the variables $\gamma(\xi)$ in $H^*$, we can find a solution with value at least $\rho^*$.

\begin{subequations} \label{eq:lp_signaling_reduced}
	\begin{align}
	\max_{\gamma \in \Delta_{H^*},z\le 0} & \,\, \sum_{\xi \in H^*} \gamma(\xi) \, \textsc{Rev}(V,\xi) + \beta z \quad \text{s.t.} \\
	& \sum_{\xi \in H^*} \gamma(\xi) \, \xi(\theta)-z \ge \mu_\theta & \forall \theta \in \Theta.
	\end{align}
\end{subequations}

To conclude the proof, notice that the algorithm provides an $dm^2 /\beta + \eta+\lambda$, where the term $dm^2 /\beta$ is due to the relaxation of the primal and  $\eta+\lambda$ to the use of an approximate separation oracle. Since the algorithm runs in time polynomial in $\beta$, $1/\eta$ and $1/\lambda$, we can provide an arbitrary good approximations choosing sufficiently small values of $\eta$ and $\lambda$, and a sufficiently large value for $\beta$.
\end{proof}

%\onecolumn
\rvFinite*
\begin{proof}
Let $\gamma$ be an optimal distribution over $\Delta_{\Theta}$ satisfying the consistency constraints. We define $\gamma^* \in \Delta_{\Xi}$ as follow:
\begin{equation*}
    \gamma^*(\tilde{\xi})=\sum_{\xi \in supp(\gamma)} \gamma(\xi)  \gamma_{\xi}(\tilde{\xi}) \hspace{4mm} \forall \tilde{\xi} \in \Xi,
\end{equation*}
where $\gamma_{\xi}$ is the distribution that satisfies Definition \ref{def:stability}.
It easy to see that $\gamma^* \in \Delta_{\Xi^*}$ satisfies the consistency constraints. Moreover, since $\textsc{Rev}(\mathcal{V},\tilde{\xi})$ is $(\alpha,\epsilon)$-stable for $\Xi$, we get:
\begin{align*} 
    \sum_{\tilde{\xi} \in \Xi}\gamma^*(\tilde{\xi}) \textsc{Rev}(\mathcal{V},\tilde{\xi}) 
    & =   \sum_{\xi \in supp(\gamma)} \gamma(\xi) \sum_{\tilde{\xi} \in \Xi }
   \gamma_\xi(\tilde{\xi})  \textsc{Rev}(\mathcal{V},\tilde{\xi})  \\
     & \ge  \sum_{\xi \in supp(\gamma)} \gamma(\xi) ( (1-\alpha) \textsc{Rev}(\mathcal{V},\xi) - \varepsilon)\\
     & =   (1-\alpha) \sum_{\xi \in supp(\gamma)}  \gamma(\xi) \textsc{Rev}(\mathcal{V},\xi)  - \varepsilon \\
     & =   (1-\alpha) OPT  -\varepsilon
\end{align*}
\end{proof}

\rvSamples*
\begin{proof}
We first observe that for each $\xi \in \Delta_{\Theta}$ we have:
\begin{equation*}
\textsc{Rev}(\mathcal{V}, \xi) \le \sum_{j=1}^{m} j \ (\lambda_{j}-\lambda_{j+1}) \le  m \sum_{j=1}^{m}  \ (\lambda_{j}-\lambda_{j+1})  = \lambda_{1} m
\end{equation*}
So that by Hoeffding bound we have: 
\begin{equation*}
Pr \Big( \hspace{0.3mm} \big|  \textsc{Rev}(\mathcal{V}_s, \xi)- \textsc{Rev}( \mathcal{V}, \xi) \hspace{0.3mm} \big | \le \tau/2 \Big ) \ge 1- 2 \hspace{0.2mm}e^{\frac{- s \tau^2}{2 (\lambda_{1} m)^2}} = 1-\rho
\end{equation*}
where the inequality is attained considering a number of samples $s = \frac{2 (\lambda_{1}m )^2}{ \tau^2} \log(\frac{2}{\rho})$. Moreover, by union bound and De Morgan's laws, we get:
\begin{equation*}
Pr \Big( \bigcap_{\xi \in \Xi} \Big\{ \hspace{0.3mm} \big| \textsc{Rev}(\mathcal{V}_s, \xi) - \textsc{Rev}(\mathcal{V}, \xi) \hspace{0.3mm} \big | \le \tau /2 \Big \} \Big ) \ge 1- \rho |\Xi|
\end{equation*}
Let $\gamma^*\in \Delta_{\Xi} $ be an optimal solution of the problem when the auctioneer can observe the actual receivers' valuation distribution while let $\gamma_{\mathcal{V}_s} \in \Delta_{\Xi}$ an optimal solution of the problem when the distribution is the empirical one. We observe that the expected revenue provided by $\gamma_{\mathcal{V}_s} \in \Delta_{\Xi}$ will be greater or equal to the one obtained with $\gamma^*\in \Delta_{\Xi} $ minus a fixed parameter with a probability of at least $1-\rho|\Xi|$. Formally it holds:
\begin{equation*}
\sum_{\xi \in \Xi} \gamma_{\mathcal{V}_s}(\xi) \textsc{Rev}( \mathcal{V}, \xi) \ge \sum_{\xi \in \Xi } \gamma_{\mathcal{V}_s}(\xi) \textsc{Rev}(\mathcal{V}_s, \xi) - \tau/2 \ge \sum_{\xi \in \Xi} \gamma^*(\xi) \textsc{Rev}(\mathcal{V}_s, \xi) - \tau/2 \ge \sum_{\xi \in \Xi} \gamma^*(\xi) \textsc{Rev}( \mathcal{V}, \xi) - \tau
\end{equation*}
Finally, we indicate with $\mathbb{E} [\sum_{\xi \in \Xi}\gamma_{\mathcal{V}_s}(\xi) \textsc{Rev}(\mathcal{V}, \xi) ]$ the expectation over the sampling procedure of the expected revenue archived by a solution of the LP considering the empirical receivers distribution. Finally, we have that:
\begin{align*}
\mathbb{E} [ \sum_{\xi \in \Xi}\gamma_{\mathcal{V}_s} (\xi) \textsc{Rev}(\mathcal{V}, \xi) ] & \ge \Big ( 1- \rho |\Xi| \Big ) \mathbb{E}\Big[ \sum_{\xi \in \Xi }   \gamma_{\mathcal{V}_s} (\xi) \textsc{Rev}( \mathcal{V}, \xi) \hspace{0.2mm} | \bigcap_{\xi \in \Xi} \Big\{ \hspace{0.3mm} \big| \textsc{Rev}( \mathcal{V}_s, \xi)- \textsc{Rev}(\mathcal{V}, \xi) \hspace{0.3mm} \big | \le \tau/2 \Big \} \Big ] \\
& \ge \Big ( 1- \rho |\Xi| \Big ) \sum_{\xi \in \Xi} \gamma^*(\xi) \textsc{Rev}( \mathcal{V}, \xi) -  \tau \\
& \ge \Big ( 1- \rho |\Xi| \Big ) OPT_{\Xi} -  \tau
\end{align*}
\end{proof}

\UnkownFixedKLem*
\begin{proof}
	
	We show that there exists a distribution $\gamma \in \Delta_{\Xi^q}$ over $q$-uniform posterior that provides an expected revenue that satisfies the conditions in the statement.
	For a $\xi \in \Delta_\Theta$, let $\xi^q \in \Xi^q$ be the empirical mean of $q$ vectors built form $q$ {i.i.d.} samples drawn from the given posterior $\xi$.
	In particular, each sample is obtained by randomly drawing a state of nature, with each state $\theta \in \Theta$ having probability $\xi(\theta)$ of being selected, and, then, a $d$-dimensional vector is built by letting all its components equal to $0$, except for that one corresponding to $\theta$, which is set to $1$.
	Notice that $\xi^q$ is a random vector supported on $q$-uniform posteriors, whose expected value is posterior $\xi$.
	Then, we let $\gamma_\xi \in \Delta_{\Xi^q}$ be such that, for every $\tilde \xi \in \Xi^q$, it holds $\gamma_{\xi}(\tilde \xi) = \textnormal{Pr} \left\{  \xi^q = \tilde \xi \right\}$.
	% It is easy to check that the total number of possibly different sample distributions built considering $q$ samples will be at most $d^q$.
	%We will now introduce a new signaling scheme $\gamma \in \Delta_{\Xi_q}$ in a way that: $\gamma(\tilde{\xi})  = Pr(\tilde{\xi} = \xi^q )\hspace{1.5mm}  \forall \tilde{\xi} \in \Xi_q$.\\
	 In the following, given a $\xi \in \Delta_{\Theta}$ and a $j \in [m+1]$ we write $\xi^\top v_{i_j}$ without specifying that $\pi=(i_1,...,i_{m+1}) \in \Pi_{m+1}$ is the vector such that $\xi \in \Xi_\pi$. Then, by Hoeffding bound we have that:
	\begin{equation*}
	Pr \Big( \tilde{\xi}^\top v_{i_j} \ge \xi^\top v_{i_j} -  \eta \Big) \ge  1- \hspace{0.2mm}e^{-2q \eta^2 } = 1- \frac{ \eta}{ m+1} \hspace{2mm} \forall j \in \{1,...,m+1 \} 
	\end{equation*}
	where the equality follows from the definition of $q$.
	Thanks to union bound and De Morgan's laws we get:
	\begin{equation*}
	Pr \Big( \bigcap_{j=1}^{m+1} \hspace{0.2mm} \{ \tilde{\xi}^\top v_{i_j} \ge \xi^\top v_{i_j} - \varepsilon \} \Big ) \ge  1-\eta, 
	\end{equation*}
	%where $\pi=(i_1,\dots,i_n) \in \Pi_n$ is such that $\xi \in \Pi_n$.
	Now, we prove that the revenue is $(0,2\eta m)$-stable over $\Xi_q$.
	%In particular, it holds:
	%\textcolor{red}{GLI INDICI i DI $v_i$ DOVREBBERO DIPENDERE DA $\xi$, COSI FORMALMENTE NON VOGLIONO DIRE NIENTE. ABBIAMO UNA NOTAZIONE?}
	%\textcolor{purple}{ Given $\xi$ there always exists a $\pi = (i_1, ..., i_{m+1}) \in \Pi_{m+1}$ such that $\xi \in \Xi_{\pi}$  so that 
	%$\mathbb{E}_{\tilde{\xi} \sim \gamma_{\xi}} \Big[ \textsc{Rev}(\mathcal{V}, \tilde{\xi} ) \hspace{0.2mm} \Big]  \ge \mathbb{E}_{\tilde{\xi} \sim \gamma} \Big[ \textsc{Rev}(\mathcal{V}, \tilde{\xi} ) \hspace{0.2mm} \big | \bigcap_{j=1}^m \hspace{0.2mm} \{ \tilde{\xi}  v_{i_j} \ge \xi v_{i_j} -  \eta \}  \Big] Pr \Big( \bigcap_{j=1}^{m+1} \hspace{0.2mm} \{ \tilde{\xi}  v_{i_j} \ge \xi v_{i_j} - \eta \} \Big )$ }

	\begin{align*}
	\mathbb{E}_{\tilde{\xi} \sim \gamma_{\xi}} \Big[ \textsc{Rev}(\mathcal{V}, \tilde{\xi} ) \hspace{0.2mm} \Big] & \ge \mathbb{E}_{\tilde{\xi} \sim \gamma_\xi} \Big[ \textsc{Rev}(\mathcal{V}, \tilde{\xi} ) \hspace{0.2mm} \big | \bigcap_{j=1}^{m+1} \hspace{0.2mm} \{ \tilde{\xi}^\top v_{i_j} \ge \xi^\top v_{i_j} -  \eta \}  \Big] Pr \Big( \bigcap_{j=1}^{m+1} \hspace{0.2mm} \{ \tilde{\xi}^\top v_{i_j} \ge \xi^\top v_{i_j} - \eta \} \Big ) \\
	& \ge (1-  \eta) \mathbb{E}_{\tilde{\xi} \sim \gamma_{\xi}} \Big[ \textsc{Rev}(\mathcal{V}, \tilde{\xi} ) \hspace{0.2mm} | \bigcap_{j=1}^{m+1} \hspace{0.2mm} \{ \tilde{\xi}^\top v_{i_j} \ge \xi^\top v_{i_j} - \eta \}  \Big] \\
	& \ge  (1-  \eta)   \Big(\sum_{j=1}^{m} \hspace{0.5mm} j\hspace{0.5mm} (\lambda_{j}-\lambda_{j+1}) (\xi^\top v_{i_{j+1}} - \eta  ) \Big) \\
	& \ge  (1-  \eta) \big( \textsc{Rev}(\mathcal{V}, \xi) - \eta m  \big) \\
	& \ge \textsc{Rev}(\mathcal{V}, \xi) - 2 \eta m  
	\end{align*}
	Proving that the revenue is $(0,2\eta m)$-stable over $\Xi_q$ with $q \ge \frac{1}{2 \eta^2}\log \big(\frac{m+1}{\eta }\big)$. Hence, by Lemma \ref{lem:rvFinite}, $OPT_{\Xi_q} \ge OPT - 2 \eta m $ proving the first point of the lemma.

	Now, we prove that the revenue is $((1-\frac{\eta}{\delta})^2, 0)$-stable over $\Xi_q$.
	In particular, it holds:
	\begin{align*}
	\mathbb{E}_{ \tilde{\xi} \sim \gamma_{\xi} } \Big[ \textsc{Rev}(\mathcal{V}, \tilde{\xi} ) \hspace{0.2mm} \Big] & \ge
	\mathbb{E}_{\tilde{\xi} \sim \gamma_{\xi} } \Big[ \textsc{Rev}(\mathcal{V}, \tilde{\xi} ) \hspace{0.2mm} \big | \bigcap_{j=1}^{m+1}  \hspace{0.2mm} \{ \tilde{\xi}^\top v_{i_j} \ge \xi^\top v_{i_j} -  \eta \}  \Big] Pr \Big( \bigcap_{j=1}^{m+1}   \hspace{0.2mm} \{ \tilde{\xi}^\top v_{i_j} \ge \xi^\top v_{i_j} - \eta \} \Big ) \\
	& \ge (1-  \eta) \mathbb{E}_{ \tilde{\xi} \sim \gamma_{\xi} } \Big[ \textsc{Rev}(\mathcal{V}, \tilde{\xi} ) \hspace{0.2mm} | \bigcap_{j=1}^{m+1}  \hspace{0.2mm} \{ \tilde{\xi}^\top v_{i_j} \ge \xi^\top v_{i_j} - \eta \}  \Big] \\
	& \ge (1- \eta )  \Big(   \sum_{j=1}^{m} \hspace{0.5mm} j\hspace{0.5mm} (\lambda_{j}-\lambda_{j+1}) (\xi^\top v_{i_{j+1}} - \frac{\eta }{\delta} \xi^\top v_{i_{j+1}} ) \Big) \\
	& \ge (1- \eta ) \Big(  \Big(1- \frac{\eta}{\delta} \Big)  \sum_{j=1}^{m} \hspace{0.5mm} j\hspace{0.5mm} (\lambda_{j}-\lambda_{j+1}) \xi^\top v_{i_{j+1}} \Big) \\
	& \ge \Big( 1- \frac{\eta}{\delta} \Big)^2 \textsc{Rev}(\mathcal{V}, \xi) 
	\end{align*}
	Proving that the revenue is $((1-\frac{\eta}{\delta})^2, 0)$-stable over $\Xi_q$ with $q \ge \frac{1}{2 \eta^2}\log \big(\frac{m+1}{\eta }\big)$. Thus, by Lemma \ref{lem:rvFinite}, $OPT_{\Xi_q} \ge (1-\frac{\eta}{\delta})^2 OPT $ proving the second point of the lemma.
\end{proof}

\UnkownFixedDLem*
\begin{proof}
First, we show that the revenue is a Lipschitz continuous function in the posterior probability $\xi \in \Delta_{\Theta}$ with respect to the infinity  norm. In particular, it holds:
\begin{equation*}
|\textsc{Rev}(\mathcal{V}, \xi)-\textsc{Rev}(\mathcal{V},\xi')| \le m d ||\xi-\xi'||_{\infty} \hspace{3mm}\forall \xi,\xi' \in \Delta_{\Theta}
\end{equation*}
This follows from the fact that $\textsc{Rev}(\mathcal{V},\xi)$ is a piecewise linear, continuous function and the partial derivative of $\textsc{Rev}(\mathcal{V}, \xi)$ with respect each component of $\xi$ is bounded almost every where by $m$.
%To prove the latter inequality we parametrize the revenue as a function of a single variable, in particular setting $v=\xi-\xi'$ and $l= || \xi-\xi'||_2 $ we will have that $Rev(x,\mathcal{V}) = Rev(\xi + x v, \mathcal{V})$ $\forall x \in [0, l]$. We also observe that $Rev(x,\mathcal{V})$ is a piecewise linear, continuous and thus absolutely continuous function, moreover, since its derivative is bounded almost everywhere by $k$, (i.e. $|\frac{\partial}{\partial x} Rev(x,\mathcal{V}) | \le k$ $a.e.$) it will be also a Lipschitz continuous function for all $x \in [0,l]$. We still have to show that $Rev(\xi, \mathcal{V})$ is a Lipschitz continuous function too for all $\xi \in \Xi$. In order to do that we will consider the following inequality:
%\begin{align*}
    %|Rev(\xi,\mathcal{V})-Rev(\xi',\mathcal{V})| & =|Rev(l,\mathcal{V})-Rev(0,\mathcal{V})| \\
    % & \le k || \xi-\xi'||_2 \\
    % & \le k d|| \xi-\xi'||_{\infty}
%\end{align*}
%since the latter result holds for all $\xi,\xi' \in \Xi$ the revenue results a Lipschitz continuous function with respect the infinity norm.
 Then, we show that each posterior $\xi$ can be decomposed in a probability distribution $\gamma_\xi \in \Delta_{\Xi^q}$ with a small loss of revenue.
We define $I_{\lambda}(\xi)=\{ \xi' \in \Delta_{\Theta} \hspace{2mm}|\hspace{2mm} || \xi-\xi'||_{\infty} \le \lambda/md \}$ as the neighbourhood of a given posterior $\xi \in \Delta_{\Theta}$ and $\Xi(\xi)= I_{\lambda} (\xi)\cap \Xi_q $ its intersection with the set $\Xi_q$.
It is easy to see that $\xi \in co(\Xi(\xi))$. 
Hence, by Caratheodory's theorem we can decompose each $\xi$ as follow:
\begin{equation*}
    \sum_{\tilde{\xi} \in \Xi(\xi) } \gamma_\xi(\hspace{0.2mm}\tilde{\xi}\hspace{0.2mm})\tilde{\xi}(\theta) = \xi(\theta) \hspace{3mm} \forall \theta \in \Theta
\end{equation*}
with $\gamma_\xi \in \Delta_{\Xi(\xi)}$. We show now that such a decomposition will decrease the expected revenue under $\xi \in \Delta_{\Theta}$ of at most a fixed parameter. Formally, we have that:

\begin{align*}
   \mathbb{E}_{\tilde{ \xi} \sim \gamma_{\xi}} \Big[  \textsc{Rev}(\mathcal{V}, \tilde{\xi})  \Big] & =
     \sum_{\tilde{\xi} \in {\Xi(\xi)} } \gamma_\xi(\tilde{\xi}) \textsc{Rev}(\mathcal{V}, \tilde{\xi})  \\
   & \ge   \sum_{\tilde{\xi} \in \Xi(\xi)} \gamma_{\xi}(\tilde{\xi}) (\textsc{Rev}(\xi,\mathcal{V})-\lambda)   \\
   & = \textsc{Rev}(\mathcal{V}, \xi) - \lambda,
\end{align*}
where the inequality comes from the Lipschitz continuity of $\textsc{Rev}(\mathcal{V},\xi)$ and $|| \xi-\tilde \xi||_{\infty} \le \lambda/md $ for all $\tilde \xi \in \Xi(\xi)$.
This proves that the revenue is $(0 , \lambda)$-stable over $\Xi_q$. By lemma \ref{lem:rvFinite} we have that $OPT_{\Xi_q} \ge OPT-\lambda $
\end{proof}

\UnkownFixedDThm*
\begin{proof}
Let $\eta$ be the desired approximation and let $\tau$, $\alpha$, and $q$ be three suitable values defined in the following.
Applying Lemma \ref{lem:rvSamples} for $\Xi=\Xi_q$, $\rho=\alpha/m$, $\tau$ and $s=\frac{2 m^2}{ \tau^2} \log(\frac{2m}{\alpha})$, we get:
\begin{align*}
 \mathbb{E}  \Big[  \sum_{\xi \in \Xi_q} \gamma_{\mathcal{V}_s}(\xi) \textsc{Rev}(\mathcal{V}, \xi) \Big]
 & \ge \Big ( 1- \frac{\alpha |\Xi_q|}{m}  \Big ) OPT_{\Xi_q} -  \tau \\
 & \ge OPT_{\Xi_q} -  \tau - \alpha |\Xi_q|.
\end{align*}
 By Lemma \ref{lem:UnkownFixedDLem}, for a value $\lambda$ defined in the following and $q=\lceil \frac{md}{\lambda} \rceil$  we have that:
\begin{align*}
\mathbb{E} \Big [\sum_{\xi \in \Xi_q} \gamma_{\mathcal{V}_s}(\xi) \textsc{Rev}(\mathcal{V}, \xi) \Big ] & \ge OPT_{\Xi_q} -  \tau- \alpha |\Xi_q|\\
 & \ge OPT - \lambda - \tau- \alpha |\Xi_q| \\
 & = OPT - \eta,
\end{align*}
where the last equality holds taking $\lambda = \eta/3 $, $\varepsilon = \eta/6 $, $\alpha = \eta/(3|\Xi_q|)$ 
\end{proof}

\UnkownFixedKThm*
\begin{proof}
Let $\nu$ be the desired approximation and let $\tau$, $\alpha$, and $q$ be three suitable values defined in the following.
Applying Lemma \ref{lem:rvSamples} for $\Xi=\Xi_q$, $\rho=\alpha/m$, $\tau$ and $s=\frac{2 m^2}{\tau^2} \log\frac{2 m}{\alpha}$, we get:
\begin{align*}
 \mathbb{E}  \Big[  \sum_{\xi \in \Xi_q} \gamma_{\mathcal{V}_s}(\xi) \textsc{Rev}(\mathcal{V}, \xi) \Big] & \ge \Big ( 1- \frac{\alpha |\Xi_q|}{m}  \Big ) OPT_{\Xi_q} - \tau \\
 & \ge OPT_{\Xi_q} -  \tau - \alpha |\Xi_q|.
\end{align*}
 By Lemma \ref{lem:UnkownFixedKLem}, we will have that for a value $\eta$ defined in the following and $q=\frac{1}{2\eta^2}\log \frac{m+1}{\eta}$ it holds:
\begin{align*}
 \mathbb{E} \Big [\sum_{\xi \in \Xi_q} \gamma_{\mathcal{V}_s}(\xi) \textsc{Rev}(\mathcal{V}, \xi) \Big ]  & \ge  OPT_{\Xi_q}  -  \tau - \alpha |\Xi_q| \\
& \ge OPT - 2 \eta m -  \tau - \alpha |\Xi_q| \\
& = OPT - \nu
\end{align*}
Where the last equality holds for $\eta = \nu/6m $, $\tau = \nu/3 $, and $\alpha = \nu/3|\Xi_q|$. 
\end{proof}

\UnkownFixedKHigh*
\begin{proof}
Let $\beta$ be the desired approximation.
Moreover, let $\alpha$ and $\eta$ be values defined in the following, $\tau = \nu \delta \lambda_1 $, $q = \frac{1}{2 \eta^2}\log \big(\frac{m+1}{\eta }\big)$, and $s = \frac{2 m^2}{( \delta \nu)^2} \log(\frac{2}{\alpha})$.
By Lemma \ref{lem:rvSamples}, it holds
\begin{align*}
   \mathbb{E}\Big[ \sum_{\xi \in \Xi_q} \gamma_{\mathcal{V}_s}(\xi) \textsc{Rev}(\mathcal{V}, \xi) \hspace{0.2mm} \Big] 
   & \ge \Big ( 1- \alpha |\Xi_q| \Big ) OPT_{\Xi_q}  - \nu \delta \lambda_1  \\
   & \ge \Big (1 - \nu - \alpha |\Xi_q|\Big ) OPT_{\Xi_q} .
\end{align*}
By Lemma \ref{lem:UnkownFixedKLem} we have:
\begin{align*}
 \mathbb{E} \Big [\sum_{\xi \in \Xi_q} \gamma_{\mathcal{V}_s}(\xi) \textsc{Rev}( \mathcal{V}, \xi) \Big ]  & \ge \Big( 1- \nu - \alpha |\Xi_q| \Big) OPT_{\Xi_q}   \\ & \ge ( 1- \nu - \alpha |\Xi_q| ) \Big(1 - \frac{\eta}{\delta} \Big)^2 OPT \\ & = (1-\beta)OPT,
\end{align*}
where the last inequality holds for $\nu=\beta/4$, $\alpha=\frac{\beta}{4|\Xi_q|}$, and $\eta=\delta \beta/2$.
\end{proof}

\UnkownFixedKHighTight*
\begin{proof}
We reduce from public signaling in elections with a k-voting rule.
In particular, each receiver $i \in \N$ has an utility $u^i_\theta \in [-1,1]$ in a state $\theta$, where $u^i_\theta$ represent the difference between the utility of receiver $i$ in voting $c_0$ with respect to $c_1$. 
The sender's utility is $1$ if the at least $k$-voters vote for $c_0$, \emph{i.e.}, the induced posterior $\xi \in \Delta_\Theta$ is such that  $\sum _{\theta} \xi(\theta) u^i_\theta \ge0$. Otherwise, the sender's utility is $0$. See~[\citeauthor{castiglioni2020public},~\citeyear{castiglioni2020public}] for a more detailed description of the problem.
~\citeauthor{castiglioni2020public}~[\citeyear{castiglioni2020public}] show that assuming the ETH, there exists a constant $\epsilon>0$ such that distinguish between this two cases requires $n^{\tilde \Omega (log(n))}$ time:
\begin{itemize}
	\item there exists a signaling scheme such that in any induced posterior $\xi$ at least $k$ voters have $\sum_{\theta} \xi(\theta) u^r_\theta\ge 0$;
	\item in all the posteriors $\xi$ there are strictly less than k receiver with $\sum_{\theta} \xi(\theta) u^i_\theta \ge  -\epsilon $.
\end{itemize}

To prove the theorem, we show how to reduce the $k$-voting problem to our revenue maximization problem.
In particular, given an instance of $k$-voting, we build an instance of signaling with the same number of receivers.
The valuation of receiver $i$ in a state $\theta$ is $v_i(\theta)=\frac{u^i_\theta}{3}+\frac{2}{3}$,
Moreover, there are $m=k-1$ slots with $\lambda_{j}=1$ for each $j\in [m]$.
Finally, we set the required approximation $\omega=\epsilon/2$.

We show that when the first case holds, the revenue is at least $(k-1)\frac{2}{3}$, while in the second case it is strictly less than $(k-1)\frac{2-\epsilon}{3}$. Hence, a $\frac{2-\epsilon}{3}/\frac{2}{3}= 1-\epsilon/2=1-\omega$ approximation to the signaling problem can be used to provide an $\epsilon$ approximation to $k$-voting.
Since we provide a polynomial time reduction from k-voting to the revenue maximization problem, this is sufficient to prove the theorem.

\textbf{soundness.} Suppose that there exists a signaling scheme such that in any induced posterior $\xi$ at least $k$ voters have $\sum_{\theta} \xi(\theta) u^i_\theta\ge 0$. 
Consider the same signaling scheme in the revenue maximization problem. Then, in all the induced posteriors $\xi$ there are at least $k$ receivers with expected valuation at least $\frac{2}{3}$ and the total revenue is at least $(k-1)\frac{2}{3}$.

\textbf{completeness.} Suppose that in all the posteriors $\xi$ there are strictly less than k receiver with $\sum_{\theta} \xi(\theta) u^i_\theta \ge  -\epsilon$.
Notice that the revenue of a posterior is given by $(k-1)x$, where x is the k-th largest expected valuation. 
Hence, the maximum revenue is strictly less than $(k-1)\frac{2-\epsilon}{3}$.
\end{proof}

\end{document}